\documentclass[journal]{IEEEtran}

\usepackage{cite}
\usepackage{empheq}
\usepackage{amsmath}
\usepackage{amsthm}
\usepackage{amssymb}
\usepackage{mathrsfs}
\usepackage{ulem}
\usepackage{enumerate}
\usepackage{graphicx}
\graphicspath{{Figs/}{../pdf/}{../jpeg/}}
\DeclareGraphicsExtensions{.pdf,.jpeg,.png}

\usepackage{cancel}
\usepackage{makecell}
\usepackage{bbding}
\usepackage{booktabs}
\usepackage{array}

\newtheorem{theorem}{Theorem}[section]
\newtheorem{lemma}[theorem]{Lemma}

\newtheorem{corollary}[theorem]{Corollary}
\newtheorem{proposition}[theorem]{Proposition}

\newtheorem{example}{Example}

\usepackage{xcolor}

\newcommand{\tr}[0]{\color{black}}

\begin{document}

\title{Gain and Phase: Decentralized Stability Conditions for Power Electronics-Dominated Power Systems}
%
\author{\vspace{3mm} Linbin Huang, Dan Wang, Xiongfei Wang, Huanhai Xin, Ping Ju, Karl H. Johansson, and Florian D{\"o}rfler
\vspace{-3mm}
\thanks{Linbin Huang and Dan Wang contributed equally to this work.}
\thanks{Linbin Huang and Florian D{\"o}rfler are with the Department of Information Technology and Electrical Engineering at ETH Z{\"u}rich, Switzerland. (Emails: \text\{linhuang, dorfler\}@ethz.ch)}
\thanks{Dan Wang, Xiongfei Wang, and Karl H. Johansson are with the School of Electrical Engineering and Computer Science, KTH Royal Institute of Technology, Sweden. (Emails: \text\{danwang, xiongfei, kallej\}@kth.se)}
\thanks{Huanhai Xin and Ping Ju are with the College of Electrical Engineering, Zhejiang University, China. (Emails: \text\{xinhh, pju\}@zju.edu.cn)}
\thanks{This research was supported by the National  Natural Science Foundation of China (U2066601, U22B6008, U2166204), SNSF under NCCR Automation, Swedish Research Council Distinguished Professor Grant 2017-01078,
Knut and Alice Wallenberg Foundation Wallenberg Scholar Grant.}

}
\vspace{-5mm}
	
	\maketitle
	
	\begin{abstract}
	This paper proposes decentralized stability conditions for multi-converter systems based on the combination of the small gain theorem and the small phase theorem. Instead of directly computing the closed-loop dynamics, e.g., eigenvalues of the state-space matrix, or using the generalized Nyquist stability criterion, the proposed stability conditions are more scalable and computationally lighter, which aim at evaluating the closed-loop system stability by comparing the individual converter dynamics with the network dynamics in a decentralized and open-loop manner. Moreover, our approach can handle heterogeneous converters' dynamics and is suitable to analyze large-scale multi-converter power systems that contain grid-following (GFL), grid-forming (GFM) converters, {\tr and synchronous generators}. Compared with other decentralized stability conditions, e.g., passivity-based stability conditions, the proposed conditions are significantly less conservative and can be generally satisfied in practice across the whole frequency range.

	\end{abstract}
	
	\begin{IEEEkeywords}
	Decentralized stability conditions, grid-forming control, grid-following control, power converters, power systems, small gain theorem, small phase theorem, small signal stability.
	\end{IEEEkeywords}
	
	\section{Introduction}\label{sec:intro}
	
	Power electronics converters play a significant role in modern power systems, acting as the interfaces between the power grid and renewable energy sources, high-voltage DC transmission systems, smart loads, energy storage systems, etc. The large-scale integration of power converters is changing the power system dynamics, as they have distinct dynamics compared with conventional synchronous generators~\cite{milano2018foundations}. Under such a background, new stability problems are emerging, and analyzing the stability of systems integrated with multiple power converters is essential for ensuring the secure operation of power systems~\cite{hatziargyriou2020definition}. In this paper, we focus on the small-signal stability of multi-converter systems. The small-signal stability analysis of power converters has been an important and popular topic for many years, due to the complicated dynamics in converters caused by the interaction among filters, multiple nested control loops, and the power grid. There have been many well-known methods to evaluate the stability of power converters, such as eigenvalue analysis~\cite{pogaku2007modeling}, impedance-based analysis~\cite{sun2011impedance, wang2017unified, wen2015analysis, rygg2016modified}, small gain theorem-based analysis~\cite{huang2020h}, and passivity-based analysis~\cite{harnefors2015passivity, perez2004passivity}.

    Eigenvalue analysis is based on deriving the state-space matrix of the system, which, in the context of multi-converter systems, requires a detailed, global, and closed-loop model of the whole system. Hence, it may suffer from scalability and dimensionality problems when dealing with large-scale systems. Compared with eigenvalue analysis, impedance-based analysis offers more insights into the system dynamics in a wide frequency range. Moreover, the impedance of the power grid and converters can be measured, so black-box models can be directly used for stability assessment~\cite{haberle2023mimo}. In multi-converter systems, one may need to build the impedance network for stability analysis~\cite{liu2018oscillatory}. Nonetheless, the stability analysis relies on using the generalized Nyquist stability criterion or deriving the characteristic polynomial of the closed-loop system, which may still suffer from scalability and dimensionality problems. As a remedy, if all the converters in the system have homogeneous dynamics, one can mathematically decouple the system into small-scale subsystems, and then use state-space or impedance methods to analyze the subsystems~\cite{dong2018small, motter2013spontaneous, huang2022impacts}. For instance, Ref.~\cite{dong2018small} decouples a multi-infeed system that contains homogeneous grid-following (GFL) converters and analyzes the stability from the perspective of grid strength characterized by the generalized short-circuit ratio (gSCR).

    However, it has been widely acknowledged that GFL converters, which rely on phase-locked loops (PLLs) for grid synchronization, cannot support a power electronics-dominated power system. This is because PLL aims at tracking the grid frequency, and there must be frequency sources in the system such that GFL converters can operate in a stable way. Hence, we need the so-called grid-forming (GFM) converters. Typical GFM control methods include droop control~\cite{rocabert2012control, guerrero2010hierarchical}, virtual synchronous machines~\cite{d2013virtual}, synchronverters~\cite{zhong2010synchronverters}, virtual oscillator control~\cite{johnson2013synchronization, colombino2019global}, and so on. The coexistence of GFM and GFL converters makes the stability analysis of multi-converter systems more complicated, and currently it is not clear how to evaluate the stability of large-scale multi-converter systems in a scalable and computationally feasible fashion. 
    Passivity-based analysis can potentially be used to analyze the stability of GFM-GFL hybrid multi-converter systems in a scalable and decentralized manner, i.e., if all the converters are passive, then the interconnected multi-converter system is stable~\cite{harnefors2015passivity}, but it may lead to overly conservative results. Moreover, the converter's dynamics in the low-frequency range may not satisfy the passivity condition when the synchronization dynamics are taken into account due to, for instance, the negative resistance effect of PLL~\cite{wen2015analysis,harnefors2015passivity}. {\tr In~\cite{vorobev2019decentralized}, decentralized stability certificates were developed based on the celebrated concept of dissipativity. Moreover, it was pointed out that one can use multipliers to further reduce the conservatism of the condition, which provides more flexibility. Based on the concept of ``zero exclusion principle'', a decentralized stability certificate was developed in~\cite{vorobev2022network} for DC microgrids, which introduces a topology invariant property and thus allows for ``plug-and-play'' operability of DC microgrids.}

    Recent advances in control and systems theory have extended the passivity condition by defining the phases of sectorial complex matrices, which is complementary to the definition of gains of complex matrices (i.e., singular values)~\cite{wang2020phases, chen2021phase, chen2019phase,Mao2022}. To be specific, the so-called small phase theorem, as a counterpart to the small gain theorem, was established in~\cite{chen2019phase,chen2021phase,Mao2022}, which extends the passivity condition and reduces its conservatism. Moreover, the small phase theorem can be incorporated with gain information, leading to a mixed small gain-phase theorem, which provides more flexibility in analyzing the system stability~\cite{zhao2022, woolcock2023mixed}. Ref.~\cite{huang2020h} shows that one can obtain a \textit{decentralized} stability condition based on the small gain theorem, while it may be conservative in the high-frequency range. This problem can possibly be resolved by further incorporating the small phase theorem, but it is not clear how it can fit into the setting of multi-converter systems.

    This paper proposes decentralized stability conditions for multi-converter systems based on the small gain theorem and the small phase theorem. {\tr We particularly focus on the small signal stability problem, by linearizing the system around its equilibrium point. 
    Although power systems generally contain nonlinear dynamics, the large signal stability problems are out of the scope of this paper.}
    The proposed conditions are based on splitting the system dynamics into network dynamics and individual converter dynamics and can be checked in a decentralized manner: if every converter satisfies certain conditions locally, then the interconnected closed-loop system is stable. Hence, our approach is modular, scalable, and can be used to analyze the stability of large-scale multi-converter systems. Moreover, the stability conditions are applicable to systems that contain simultaneously synchronous generators, GFM, and GFL converters. {\tr The decentralized stability conditions, in general, require full knowledge of the system. However, if the power network has a special structure, e.g., an identical R/X ratio, then we only need to know the gain (or equivalently, weighted connectivity) of the network characterized by gSCR (without knowing the detailed topology and network parameters).} 

    The rest of this paper is organized as follows: in Section~II we introduce the modeling of multi-converter systems. Section~III briefly reviews the mixed gain-phase stability criterion. Section~IV proposes the decentralized stability conditions. Section~V customizes the conditions for GFM-GFL hybrid systems. Section~VI provides simulation results. Section~VII compares our approach with passivity-based analysis and small gain analysis. Section~VIII concludes the paper.

\vspace{1mm}
\textit{Notation:} Let $\mathbb{R}$ denote the set of real numbers, $\mathbb{Z}$ denote the set of integers, and $\mathbb{C}$ the set of complex numbers. Let $\|x\|$ denote the two-norm of the vector $x$. For a matrix $A$, we use $A^*$ to denote its conjugate transpose and ${\rm det}(A)$ to denote its determinant. For a positive semi-definite matrix $A$, we use $\lambda_1(A)$ to denote its smallest eigenvalue. We use $I_n$ to denote the $n$-by-$n$ identity matrix (abbreviated as $I$ when the dimensions can be inferred from the context). We use $\otimes$ to denote the Kronecker product. We use $M = {\rm diag}(M_1, M_2,...,M_k)$ to denote the diagonal (or block-diagonal) matrix whose diagonal elements (or blocks) are $M_1, M_2,...,M_k$.

	
	\vspace{-1mm}
	
	\section{Modeling of Multi-Converter Power Systems}

 Consider a multi-converter system, where $N$ ($N \in \mathbb{Z}_{>0}$) power converters are interconnected via the power network, as illustrated in Fig.~\ref{Fig_multi_converter}. In our setting, we allow the power converters to employ different control methods and different control parameters. For instance, the converter may apply the widely-used GFL control, where a PLL is used to follow the grid frequency, as shown in Fig.~\ref{Fig_control_diagram}; one can also use GFM control which can actively establish the frequency and voltage for the power grid, e.g., by emulating the swing equation to generate the internal frequency and achieve synchronization~\cite{d2013virtual}.

\begin{figure}[!t]
	\centering
	\includegraphics[width=3.1in]{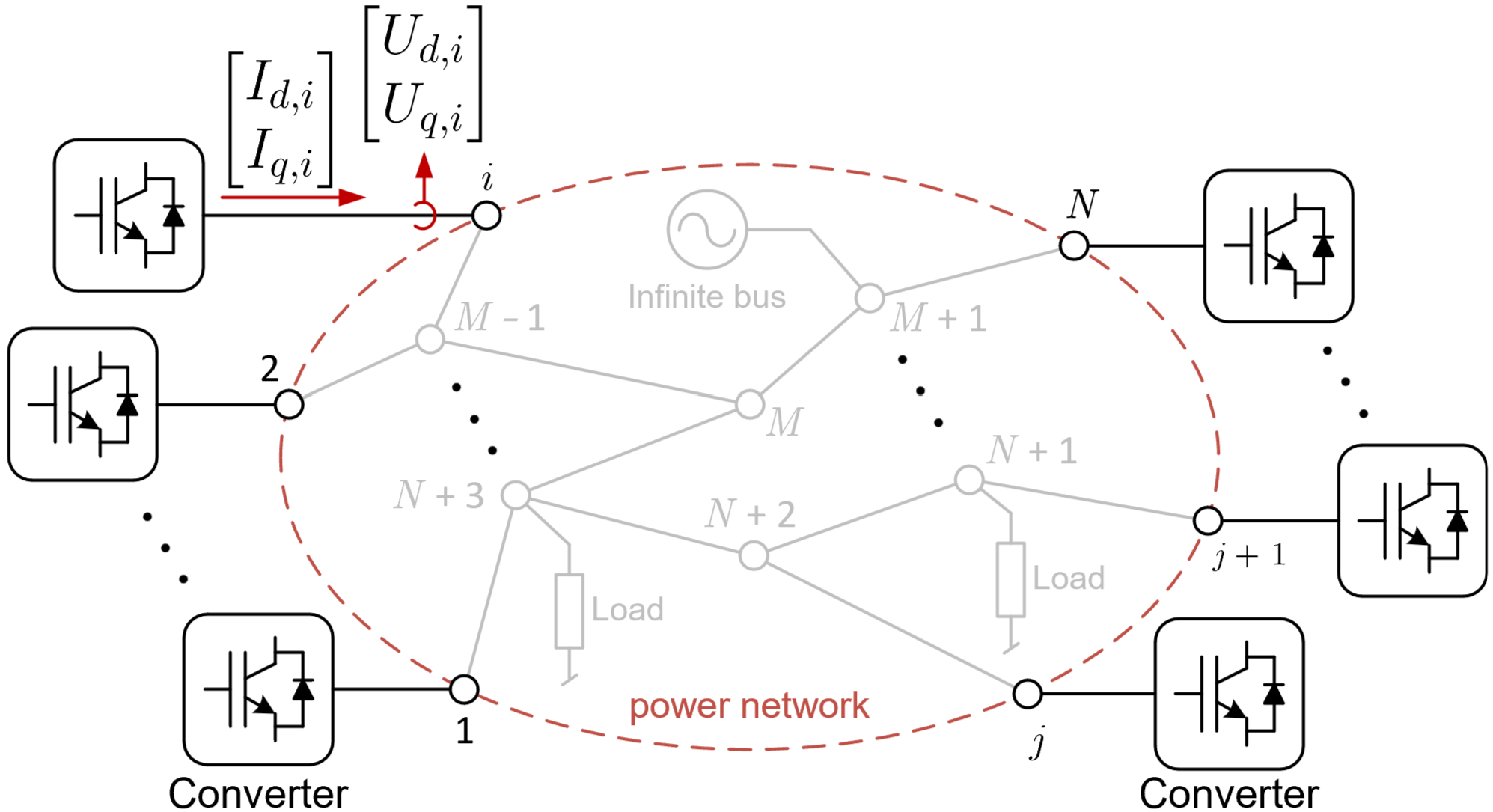}
	\vspace{-3mm}
	\caption{Illustration of a multi-converter power system.}
	\vspace{-2mm}
	\label{Fig_multi_converter}
\end{figure}

Regardless of the employed control method, the \textit{\tr small-signal dynamics} of the $i$-th converter can be compactly represented by its admittance model in the frequency domain~\cite{wang2017unified, wen2015analysis}
\begin{equation}\label{eq:admittance_model}
    -\begin{bmatrix} I_{d,i} \\ I_{q,i} \end{bmatrix} = S_i e^{J\theta_i} {\bf Y}_{{\rm C},i}(s) e^{-J\theta_i}\begin{bmatrix} U_{d,i} \\ U_{q,i} \end{bmatrix} \,,
\end{equation}
where ${\bf Y}_{{\rm C},i}(s)$ is a $2 \times 2$ transfer function matrix describing the converter's dynamics using local per-unit calculations (also known as the admittance model), 
$[ I_{d,i} \; I_{q,i} ]^\top$ and $[ U_{d,i} \; U_{q,i} ]^\top$ are respectively the output current and voltage \textit{\tr perturbations} of the \text{$i$-th} converter in a global $dq$ coordinate with a constant (nominal) rotating frequency (they correspond to ${\bf I}_{{\bf abc},i}$ and ${\bf U}_{{\bf abc},i}$ in Fig.~\ref{Fig_control_diagram}), $S_i$ is the capacity ratio of the $i$-th node’s rated capacity to the base capacity of per-unit calculation (we use the same base values when performing per-unit calculations for the global voltage and current signals).
We refer to, for instance~\cite{xin2022many} and~\cite{yang2020placing}, for the detailed derivations of ${\bf Y}_{{\rm C},i}(s)$ when GFL control or GFM control is applied. Usually, ${\bf Y}_{{\rm C},i}(s)$ is derived under a certain angle for the local $dq$ coordinate (with a constant rotating frequency), e.g., by aligning the $d$-axis with the steady-state voltage vector. We denote the \textit{steady-state angle difference} between the global $dq$ coordinate and the $i$-th converter's local $dq$ coordinate by a constant $\theta_i$, as shown in~\eqref{eq:admittance_model}, where $e^{J\theta_i} = \begin{bmatrix} \cos \theta_i & -\sin \theta_i \\ \sin \theta_i & \cos \theta_i \end{bmatrix} \vspace{1mm}$. This global $dq$ coordinate is used to ensure that all the current/voltage vectors are under the same coordinate when modeling the multi-converter system~\cite{huang2022impacts}.
Note that the dynamics of the synchronization unit (e.g., a PLL) are included in ${\bf Y}_{{\rm C},i}(s)$.
We extend~\eqref{eq:admittance_model} to include the dynamics of all the $N$ converters 
\begin{equation}\label{eq:admittance_devices}
    - \underbrace{ \begin{bmatrix} I_{d,1} \\ I_{q,1}  \\ \vdots \\ I_{d,N} \\ I_{q,N} \end{bmatrix}}_{=: {\bf I}} = 
    ({\bf S}\otimes I_2) \; e^{J {\vartheta}} \; {\bf Y}_{\rm C}^N(s) \; e^{-J {\vartheta}} \;
\underbrace{\begin{bmatrix} U_{d,1} \\ U_{q,1} \\  \vdots \\ U_{d,N} \\ U_{q,N} \end{bmatrix}}_{=:{\bf U}} \,,
\end{equation}
where ${\bf S} \!=\! {\rm diag}\{S_1,\dots,S_N\} \!\!\in\!\! \mathbb{R}^{N \times N}$ is the capacity ratio matrix, $e^{J {\vartheta}} \!\!=\!\! {\rm diag}\{e^{J\theta_1},\dots,e^{J\theta_N}\} \!\!\in\!\! \mathbb{R}^{2N \times 2N}$ is the steady-state coordinate transformation matrix,
and
${\bf Y}_{\rm C}^N(s)$ is a $2N \times 2N$ block-diagonal transfer matrix extended from ${\bf Y}_{{\rm C},i}(s)$ as
\begin{equation*}
    {\bf Y}_{\rm C}^N(s) = { \begin{bmatrix}
{\bf Y}_{{\rm C},1}(s) & & & \\
& {\bf Y}_{{\rm C},2}(s) & & \\
& & \ddots & \\
& & & {\bf Y}_{{\rm C},N}(s) \\
\end{bmatrix}} \,.
\end{equation*}

Following~\cite{huang2020h}, the small-signal dynamics of a multi-converter system can be represented by the closed-loop interaction between the devices' dynamics and the power network's dynamics, where the devices' dynamics are captured by the transfer function matrix ${\bf Y}_{\rm C}^N(s)$ in~\eqref{eq:admittance_devices}. We next develop the transfer function matrix for the power network, which may include transmission lines, loads, and infinite buses (to represent large-capacity generators or other areas).

\begin{figure}[!t]
	\centering
	\includegraphics[width=3.1in]{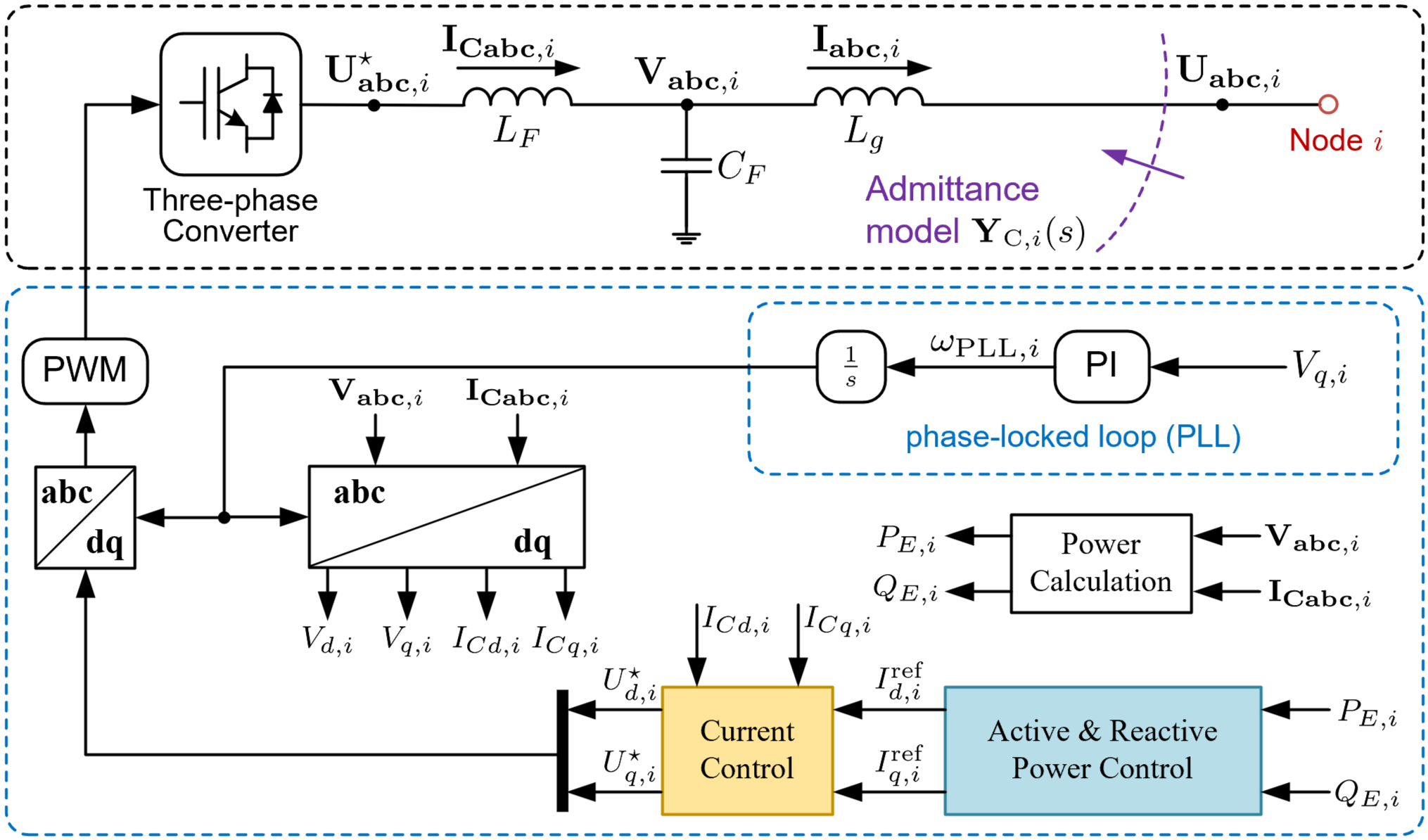}
	\vspace{-3mm}
	\caption{\tr Basic implementation of a grid-following (GFL) power converter. A more advanced PLL implementation (including voltage normalization and rated frequency feed-forward) will be considered in Example~\ref{ex:68_bus}.}
	\vspace{-3mm}
	\label{Fig_control_diagram}
\end{figure}

Consider the power network in~Fig.~\ref{Fig_multi_converter} which contains $N$ converter nodes, denoted by Nodes~$1 \sim N$, and $M - N$ interior nodes ($M \in \mathbb{Z}_{>0}$, $M \ge N$) as Nodes~$N+1 \sim M$. We consider one common grounded node, denoted by Node~$M+1$. Note that in small-signal dynamics, the infinite buses should also be considered as grounded since their voltage perturbations are zero. We consider a general setting where the dynamic equation of the line between Node~$i$ and Node~$j$ ($i,j \in \{1, \dots,M+1\}$) is
\begin{equation*}\label{eq:edge_ij}
    \begin{bmatrix} I_{d,ij} \\ I_{q,ij} \end{bmatrix} = {Y}_{ij}(s) \left( \begin{bmatrix} U_{d,i} \\ U_{q,i} \end{bmatrix} - \begin{bmatrix} U_{d,j} \\ U_{q,j} \end{bmatrix} \right) \,,
\end{equation*}
where $[ I_{d,ij} \; I_{q,ij} ]^\top$ is the current vector from Node~$i$ to Node~$j$ in the global $dq$ coordinate, $[ U_{d,i} \; U_{q,i} ]^\top$ is the voltage vector of Node~$i$ in the global $dq$ coordinate, and ${Y}_{ij}(s)$ is a $2 \times 2$ transfer function matrix encoding the dynamics of the line. For instance, if the line is composed of a resistor and an inductor (i.e., a typical power transmission line model), then we have 
\begin{equation}\label{eq:Y_ij}
    Y_{ij}(s) = B_{ij} \begin{bmatrix} s/\omega_0 + {\tr \epsilon_{ij}} & -1 \\ 1 & s/\omega_0 + {\tr \epsilon_{ij}} \end{bmatrix}^{-1} =: B_{ij} F_{\tr \epsilon_{ij}}(s) \,,
\end{equation}
where $B_{ij} = \frac{1}{X_{ij}} \vspace{0.5mm}$ is the line susceptance ($X_{ij}$ is the reactance), {\tr $\epsilon_{ij}$ is the R/X ratio of the line~$ij$}, and $\omega_0$ is the nominal frequency of the system. If there is a load (resistor) at Node~$i$, then we have (with $j = M+1$)
\begin{equation}
    Y_{ij}(s) = \begin{bmatrix} R_{{\rm Load},i}^{-1} & 0 \\ 0 & R_{{\rm Load},i}^{-1} \end{bmatrix} \,,
\end{equation}
where $R_{{\rm Load},i}$ denotes the resistance of the load; {\tr as another example, if there is a shunt capacitor (e.g., in a Pi transmission line model) at Node~$i$, then we have (with $j = M+1$)
\begin{equation*}
    Y_{ij}(s) = \begin{bmatrix} sC_i & -\omega_0 C_i \\ \omega_0 C_i & sC_i \end{bmatrix} \,,
\end{equation*}
where $C_i$ is the shunt capacitor at Node~$i$.}

We are now ready to construct the $2M \times 2M$ grounded Laplacian (transfer function) matrix ${\bf Y}(s)$ whose \textit{blocks} are 
\begin{equation}
\begin{split}
    {\bf Y}_{ij}(s) &= - Y_{ij}(s), i \ne j, \\
    {\bf Y}_{ii}(s) &= \sum_{j = 1}^{M+1} Y_{ij}(s) .
\end{split}
\end{equation}
By eliminating the $M-N$ interior nodes through Kron reduction~\cite{dorfler2012kron} one obtains a reduced network whose dynamics are captured by the following $2N \times 2N$ transfer function matrix
\begin{equation}\label{eq:Y_grid}
     {\bf I} = [\;{\bf Y}_1(s) - {\bf Y}_2(s) {\bf Y}_4^{-1}(s) {\bf Y}_3(s)\; ] {\bf U} =:{\bf Y}_{\rm grid}(s) {\bf U},
\end{equation}
where ${\bf Y}_1(s)$, ${\bf Y}_2(s)$, ${\bf Y}_3(s)$, and ${\bf Y}_4(s)$ are respectively $2N \times 2N$, $2N\times(2M-2N)$, $(2M-2N)\times 2N$, and $(2M-2N)\times (2M-2N)$ transfer function matrices defined by
\begin{equation*}
    {\bf Y}(s) =: \begin{bmatrix} {\bf Y}_1(s) & {\bf Y}_2(s) \\ {\bf Y}_3(s) & {\bf Y}_4(s) \end{bmatrix} .
\end{equation*}

Eqs.~\eqref{eq:Y_grid} and~\eqref{eq:admittance_devices} constitute the closed-loop dynamics of the multi-converter system, as shown in Fig.~\ref{Fig_closed_loop}. Note that the above modeling of the power network's dynamics is general as it can handle different dynamics in different branches. We provide an example to illustrate the structure of ${\bf Y}_{\rm grid}(s)$ below.

\begin{figure}[!t]
	\centering
	\includegraphics[width=3.2in]{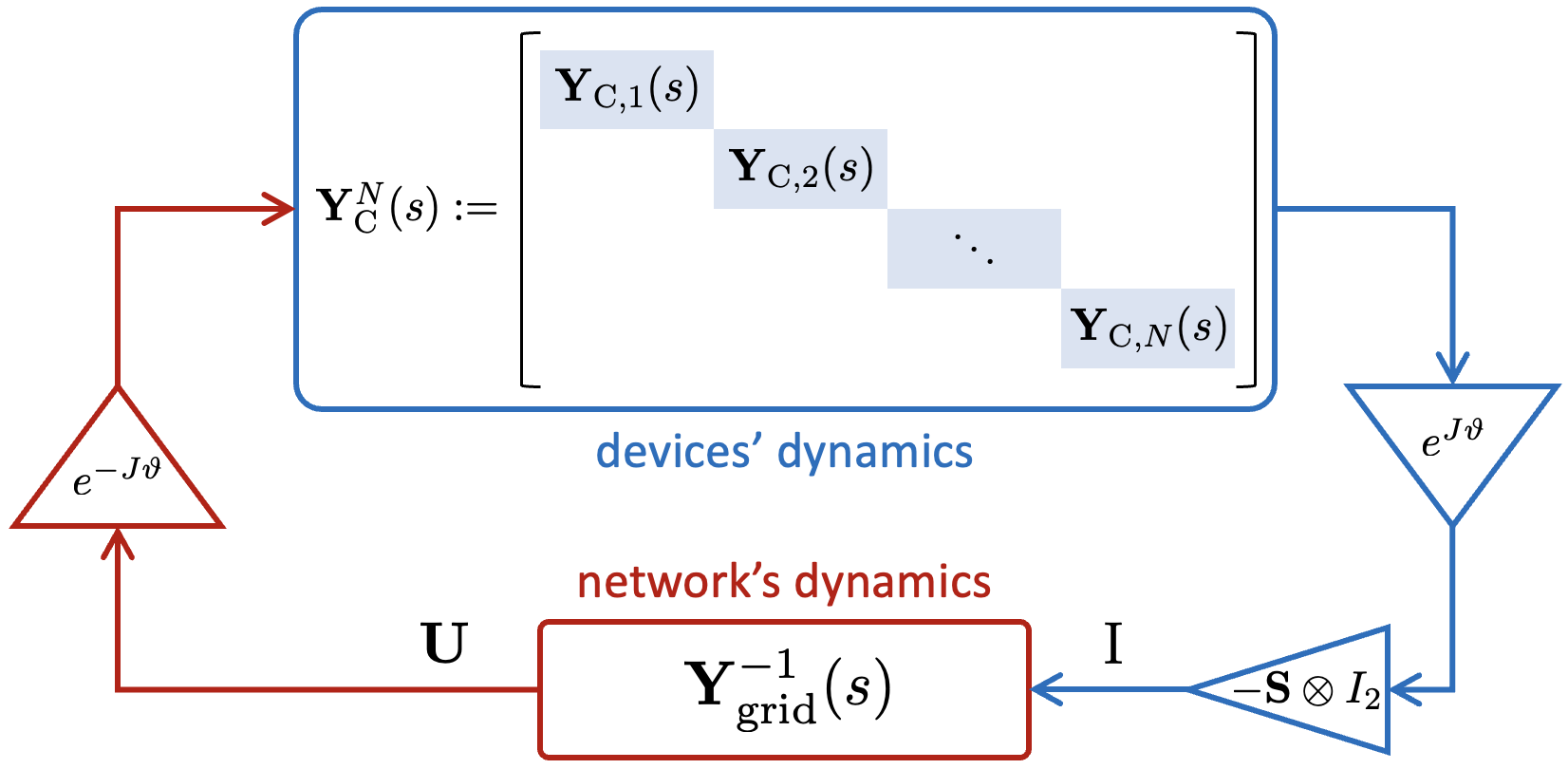}
	\vspace{-3mm}
	\caption{Closed-loop representation of a multi-{\tr device} system by combining~\eqref{eq:admittance_devices} and~\eqref{eq:Y_grid}.}
	\vspace{-2mm}
	\label{Fig_closed_loop}
\end{figure}

\begin{example}[Transmission network with an identical R/X ratio] \label{example1}
    Consider a power network {\tr with an identical R/X ratio $\epsilon$, i.e., $\epsilon_{ij} = \epsilon, \forall i,j \in \{1,...,M+1\}$, so that $Y_{ij}(s) = B_{ij}{F}_\epsilon(s)$, where 
    \begin{equation}\label{eq:F_sigma}
    {F}_\epsilon(s) = \begin{bmatrix} s/\omega_0 + {\epsilon} & -1 \\ 1 & s/\omega_0 + {\epsilon} \end{bmatrix}^{-1}.
    \end{equation}}
    Let ${\bf B} \!\in\! \mathbb{R}^{M \times M}$ be the grounded Laplacian matrix that encodes the network topology and susceptances, where ${\bf B}_{ij} = - B_{ij}$ if $i \ne j$, and ${\bf B}_{ii} = \sum_{j=1}^{M+1} B_{ij}$. By performing Kron reduction to $\bf B$, we obtain the reduced Laplacian matrix ${\bf B}_{\rm r} \in \mathbb{R}^{N \times N}$ which is calculated by ${\bf B}_{\rm r} = {\bf B}_1 - {\bf B}_2{\bf B}_4^{-1} {\bf B}_3$, where
    $$
        {\bf B}=: \begin{bmatrix} {\bf B}_1 \in \mathbb{R}^{N \times N} & {\bf B}_2 \in \mathbb{R}^{N \times (M-N)} \\ {\bf B}_3 \in \mathbb{R}^{(M-N) \times N} & {\bf B}_4 \in \mathbb{R}^{(M-N) \times (M-N)}\end{bmatrix}.
    $$
    Then, we have 
    \begin{equation}\label{eq:Ygrid_RX}
        {\bf Y}_{\rm grid}(s) = {\bf B}_{\rm r} \otimes {F}_\epsilon(s) \,.
    \end{equation}
    It can be seen that ${\bf Y}_{\rm grid}(s)$ has a special structure, {\tr i.e., it can be expressed as the Kronecker product of ${\bf B}_{\rm r}$ and ${F}_\epsilon(s)$} in this setting, which, as will be discussed later in this paper, allows us to derive simpler stability conditions compared with the general setting in~\eqref{eq:Y_grid} {\tr where the R/X ratios can be different across different lines.} {\tr In other words, our approach is general and can handle heterogeneous R/X ratios in power networks}. 
\end{example}

\vspace{0mm}

\section{Mixed Gain-Phase Stability Criterion}
This section introduces some preliminaries on the feedback stability of two linear time-invariant systems from a combined gain and phase perspective, which were recently established in \cite{chen2019phase,chen2021phase,Mao2022,zhao2022}, and will play an essential role in later developments. We first review the gains and phases of complex matrices. 

\vspace{-2mm}

\subsection{Gains and phases of complex matrices}

It has been well known that a complex matrix $A\in\mathbb{C}^{n\times n}$ has $n$ magnitudes (\textit{gains}), given by its singular values
\[
\sigma(A)=\begin{bmatrix}
    \sigma_1(A) & \sigma_2(A) & \cdots & \sigma_n(A) 
\end{bmatrix} ,
\]
arranged in such a way that
\begin{equation}
    \overline{\sigma}(A) =\sigma_1(A) \geq \dots \geq \sigma_n(A)=\underline{\sigma}(A).
    \tag{\bf Gains}
\end{equation}
By comparison, the phases of a complex matrix are less trivial concepts, which were recently defined in \cite{wang2020phases} based on the matrix's numerical range, as detailed below.
Firstly, the \textit{numerical range} of a complex matrix $A$ is defined as
\[
W(A)=\{x^*Ax: x\in \mathbb{C}^n, \|x\|=1\},
\]
which is a convex subset of $\mathbb{C}$ and contains the eigenvalues of $A$. If $0\notin W(A)$, then $A$ is said to be a \textit{sectorial matrix} \cite{wang2020phases}. For a sectorial $A$, there exists a nonsingular matrix $T$ and a diagonal unitary matrix $D$ such that 
$A=T^*DT$, which is called sectorial decomposition of $A$.
Note that $D$ is unique up to permutation and has all eigenvalues (i.e., diagonal entries) lying on an arc of the unit circle with length smaller than $\pi$. 
Then, the \textit{phases} of $A \in\mathbb{C}^{n\times n}$, denoted by 
\begin{equation}
    \overline\phi(A)=\phi_1(A) \geq \dots \geq \phi_n(A)=\underline\phi(A), 
    \tag{\bf Phases}
\end{equation}
are defined as the phases of the eigenvalues of $D$ so that $\overline\phi(A)-\underline\phi(A)<\pi$. 
Graphic illustrations of the numerical range $W(A)$ and the phases of a sectorial matrix $A$ are shown in Fig.~\ref{fig:sectorial}, where the two angles from the positive real axis to the two supporting rays of $W(A)$ are $\overline{\phi}(A)$ and $\underline\phi(A)$, respectively. The other phases of $A$ lie in between.

\begin{figure}[htb]
    \vspace{-4.5mm}
    \centering    \includegraphics[scale=0.14]{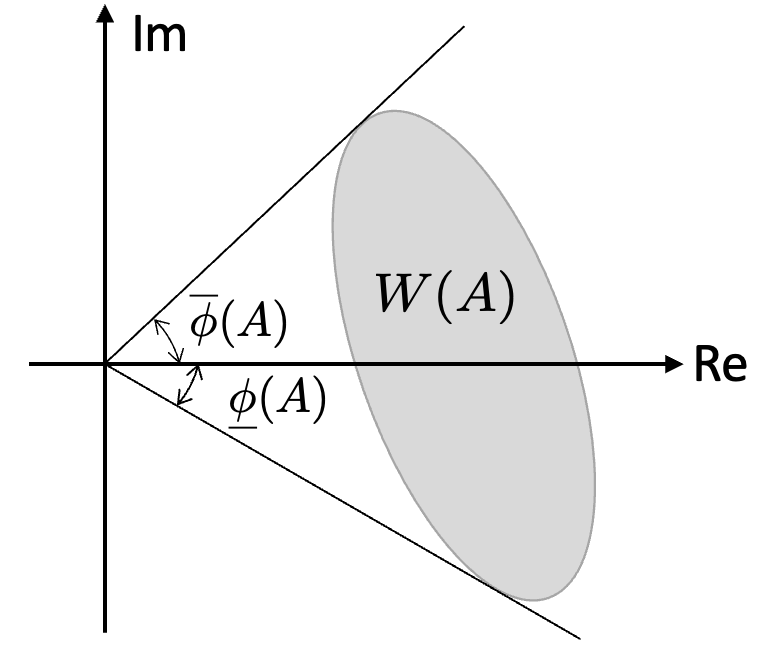}
    \vspace{-4.5mm}
    \caption{Phases of a sectorial matrix $A$ are contained between $[\underline\phi(A),\overline\phi(A)]$, and the gray area $W(A)$ is the numerical range of $A$.}
    \label{fig:sectorial}
\end{figure}

Based on the above definition, the phases of a complex matrix $A$ can be computed as follows. If $A$ is sectorial and $A = T^*DT$ is
its sectorial decomposition, then we have 
\[
A(A^*)^{-1} = T^*D^2(T^*)^{-1}.
\]
This means that the phases of $A$ (i.e., the phases of eigenvalues of $D$) are the halves of the phases of the eigenvalues of matrix
$A(A^*)^{-1}$. We refer the interested readers to~\cite{wang2020phases} for more details on the computation of the phases of a sectorial matrix.

\vspace{-2mm}
 
\subsection{Mixed small gain-phase theorem}

Let $G, H\in\mathcal{RH}_{\infty}^{m\times m}$, i.e., be $m\times m$ real rational proper stable transfer function matrices. The feedback interconnection of $G(s)$ and $H(s)$ in Fig.~\ref{fdbk} is said to
be stable if the \textit{Gang of Four matrix} (describing all input/output maps in Fig.~\ref{fdbk})
\begin{align*}
G(s)\#H(s)=\begin{bmatrix}
  (I + HG)^{-1} & (I + HG)^{-1}H\\
  G(I + HG)^{-1} & G(I + HG)^{-1}H
\end{bmatrix}
\end{align*}
is stable, i.e., $G\#H \in \mathcal{RH}^{2m \times 2m}_\infty$.

\setlength{\unitlength}{1mm}
\vspace{0mm}
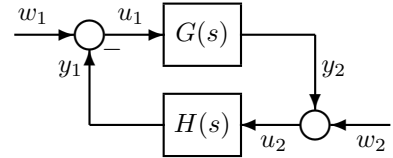
\begin{figure}[htb]
\begin{center}
\begin{picture}(50,20)
\thicklines 
\put(0,17){\vector(1,0){8}} \put(10,17){\circle{4}}
\put(12,17){\vector(1,0){8}} \put(20,13){\framebox(10,8){$G(s)$}}
\put(30,17){\line(1,0){10}} \put(40,17){\vector(0,-1){10}}
\put(38,5){\vector(-1,0){8}} \put(40,5){\circle{4}}
\put(50,5){\vector(-1,0){8}} \put(20,1){\framebox(10,8){$H(s)$}}
\put(20,5){\line(-1,0){10}} \put(10,5){\vector(0,1){10}}
\put(5,10){\makebox(5,5){$y_1$}} \put(40,10){\makebox(5,5){$y_2$}}
\put(0,17){\makebox(5,5){$w_1$}} \put(45,0){\makebox(5,5){$w_2$}}
\put(13,17){\makebox(5,5){$u_1$}} \put(32,0){\makebox(5,5){$u_2$}}
\put(10,10){\makebox(6,10){$-$}}
\end{picture}
\vspace{0mm}
\caption{A standard closed-loop system $G(s)\#H(s)$.}
\vspace{0mm}
\label{fdbk}
\end{center}
\end{figure}

The following mixed small gain-phase theorem characterizes the feedback stability, which was formulated in \cite[Theorem 2]{zhao2022}.
\begin{lemma}\label{small-gain-phase}
Let the open-loop systems $G, H\in \mathcal{RH}_{\infty}$, i.e., be real, rational, stable, and proper transfer function matrices. Then, $G(s)\# H(s)$ is stable if for each $\omega\in [0,\infty]$, either
\begin{enumerate}[i)]
    \item $\overline{\sigma}(G(j\omega)) \overline{\sigma}(H(j\omega))<1$, or
    \item $G(j\omega), H(j\omega)$ are sectorial, $\overline{\phi}(G(j\omega))\!+\!\overline{\phi}(H(j\omega))\!<\!\pi$ and $\underline{\phi}(G(j\omega))+\underline{\phi}(H(j\omega))>-\pi$.
\end{enumerate}
\end{lemma}

        \section{Decentralized Stability Conditions}
	
	We next develop decentralized stability conditions for multi-converter systems based on Lemma~\ref{small-gain-phase}. According to Fig.~\ref{Fig_closed_loop}, the closed-loop dynamics of the multi-converter system can be described as $[({\bf S}\otimes I_2)e^{J {\vartheta}} \; {\bf Y}_{\rm C}^N(s) \; e^{-J {\vartheta}}]\; \# \; {\bf Y}^{-1}_{\rm grid}(s)$, or equivalently (as proved in Appendix~\ref{Append_Proof1}), as 
    \begin{equation}\label{eq:rescaled_sys}
        [e^{J {\vartheta}} \; {\bf Y}_{\rm C}^N(s) \; e^{-J {\vartheta}} ({\bf D} \otimes {\widetilde F}_\epsilon^{-1}(s)) ]\; \# \; \widetilde{\bf Y}_{\rm grid}^{-1}(s),
    \end{equation}
    where 
    \begin{equation}\label{eq:Y_C_grid}
        \widetilde{\bf Y}_{\rm grid}(s) := ({\bf S}^{-\frac{1}{2}} \otimes I_2) {\bf Y}_{\rm grid}(s) ( {\bf S}^{-\frac{1}{2}} {\bf D} \otimes {\tr{\widetilde F}_\epsilon^{-1}(s)}),
    \end{equation}
    and ${\bf D} = {\rm diag}\{D_1,...,D_N\} \in \mathbb{R}_{>0}^{N \times N}$ is a diagonal matrix to re-scale the converters' dynamics and the network dynamics (one can simply choose ${\bf D} = I_N$), which provides more flexibility in splitting the system. {\tr The $2 \times 2$ transfer function matrix ${\widetilde F}_\epsilon(s)$ is also used to re-scale the converters' dynamics and the network dynamics, which should be chosen as ${\widetilde F}_\epsilon(s) = {F}_\epsilon(s)$ if the power network has an identical R/X ratio~$\epsilon$ (${F}_\epsilon(s)$ has been defined in~\eqref{eq:F_sigma}). In this case, the gain of the network becomes a constant which simplifies the analysis, as will be seen in the development of Corollary~\ref{prop:Power_grid_strength}. Notice that the choice of ${\widetilde F}_\epsilon(s)$ does not affect the stability of the system, as it appears on both the converters' dynamics and the network dynamics and will be canceled out in the closed loop; however, it affects the way to split the system dynamics and therefore the perspective of viewing the two open-loop systems (i.e., the converters' dynamics and the network dynamics) which may influence the conservatism of our approach. Hence, if the power network has heterogeneous R/X ratios, then one can choose 
    \[{\widetilde F}_\epsilon(s) = \begin{bmatrix} s/\omega_0 + {\widetilde \epsilon} & -1 \\ 1 & s/\omega_0 + {\widetilde \epsilon}
    \end{bmatrix}^{-1},\]
    where $\widetilde \epsilon$ can be, for instance, the average R/X ratio, or the dominant R/X ratio of the network; while in any case, unlike the case of an identical R/X ratio, the gain of the network will be a function of the frequency rather than a constant.}
    
    In short, to simplify the description of the network dynamics and to reduce the conservatism of our approach, we consider \eqref{eq:rescaled_sys} instead of $[({\bf S}\otimes I_2) e^{J {\vartheta}} \; {\bf Y}_{\rm C}^N(s) \; e^{-J {\vartheta}}]\; \# \;  {\bf Y}^{-1}_{\rm grid}(s)$. In what follows, we focus on~\eqref{eq:rescaled_sys} to represent the dynamics of a multi-converter system.
    In~\eqref{eq:rescaled_sys}, the (re-scaled) network dynamics are compactly described by $\widetilde{\bf Y}_{\rm grid}(s)$. Then, for ease of illustration of the following results, we define the (re-scaled) dynamics of the $i$-th converter as
    \begin{equation}\label{eq:Y_C_convert}
        \widetilde{\bf Y}_{{\rm C},i}(s) := D_i {\bf Y}_{{\rm C},i}(s){\widetilde F}_\epsilon^{-1}(s),\; i \in \{1,..., N\}.
    \end{equation}
    The decentralized stability conditions are presented below.
 

\begin{proposition}[Decentralized stability conditions]
\label{prop:Decentralized_stability_conditions}
The multi-converter system in Fig.~\ref{Fig_closed_loop} is stable if the open-loop systems are stable, and for each $\omega\in [0,\infty]$, \textbf{either}
\begin{enumerate}[i)]
    \item the decentralized \textbf{gain} condition, i.e.,
    \begin{equation} \label{eq:decentralized_gain}
    \hspace{-3mm}
    \max\limits_{i } \; \overline \sigma(\widetilde{\bf Y}_{{\rm C},i}(j\omega)) <  \underline \sigma( \widetilde{\bf Y}_{\rm grid}(j\omega) ) \quad {\rm \textit{holds, \textbf{or}}}
    \end{equation}
    \item the decentralized \textbf{phase} condition, i.e.,
    \begin{equation}\hspace{-3mm} 
    \begin{cases}
    \textit{a) } \max\limits_{i } \; \overline \phi(\widetilde{\bf Y}_{{\rm C},i}(j\omega)) < \pi - \overline \phi( \widetilde{\bf Y}^{-1}_{\rm grid}(j\omega) ) , \\
    \textit{b) }\min\limits_{i } \; \underline \phi(\widetilde{\bf Y}_{{\rm C},i}(j\omega)) > -\pi - \underline \phi( \widetilde{\bf Y}^{-1}_{\rm grid}(j\omega) ) , \\
    \textit{c) } \widetilde{\bf Y}_{\rm grid}(j\omega) \text{ is sectorial},\\
         \textit{and d) } \max\limits_{i } \; \overline \phi(\widetilde{\bf Y}_{{\rm C},i}(j\omega))  - \min\limits_{i } \; \underline \phi(\widetilde{\bf Y}_{{\rm C},i}(j\omega)) < \pi ,
    \end{cases} \label{eq:decentralized_phase}
    \end{equation}
    holds.
\end{enumerate}
\end{proposition}
\begin{proof}
    See Appendix~\ref{Append_Proof1}.
\end{proof}

Note that the decentralized stability conditions above require that the open-loop systems are stable, which is generally true in the power systems context because the devices are designed to be stable when directly connected to an infinite bus and the power network is passive.

Proposition~\ref{prop:Decentralized_stability_conditions} evaluates the stability of multi-converter systems in a decentralized manner. To be specific, conditions \eqref{eq:decentralized_gain} and \eqref{eq:decentralized_phase} focus on comparing the converters' characteristics with the network characteristics. For ease of interpreting Proposition~\ref{prop:Decentralized_stability_conditions}, we introduce the following terminologies: 
\begin{enumerate}[i)]
\item the \textbf{gain} curve of the $i$-th converter is the curve defined by $\overline \sigma(\widetilde{\bf Y}_{{\rm C},i}(j\omega))$ as a function of $\omega$; \vspace{1mm}
\item the \textbf{phase area} of the $i$-th converter is the area defined by $[\underline \phi(\widetilde{\bf Y}_{{\rm C},i}(j\omega)), \overline \phi(\widetilde{\bf Y}_{{\rm C},i}(j\omega))]$ as a function of $\omega$; \vspace{1mm}
\item the \textbf{gain} curve of the power network is the curve defined by $\underline \sigma( \widetilde{\bf Y}_{\rm grid}(j\omega) )$ as a function of $\omega$; \vspace{1mm}
\item the \textbf{phase area} of the power network is the area defined by $[-\pi - \underline \phi( \widetilde{\bf Y}^{-1}_{\rm grid}(j\omega) ) , \pi - \overline \phi( \widetilde{\bf Y}^{-1}_{\rm grid}(j\omega) )]$ as a function of $\omega$; \vspace{1mm}
\item the \textbf{width} of an area $[a,b]$ is $b-a$; \vspace{1mm}
\item the \textbf{sectorial transition frequency} of the $i$-th converter is the frequency at which $\widetilde{\bf Y}_{{\rm C},i}(j\omega)$ becomes sectorial.
\end{enumerate}
It can be seen from~\eqref{eq:Y_C_grid} and~\eqref{eq:Y_C_convert} that the computation of the gain curve and the phase area of the $i$-th converter requires only the information of the $i$-th converter (control scheme, parameters, etc.) and the R/X ratio of the network lines; the computation of the gain curve and the phase area of the power network requires the information of the power network (topology, line impedance, etc). In some cases, we only need to know the gSCR of the power network that characterizes its strength, as discussed in Corollary~\ref{prop:Power_grid_strength}. Note that if a matrix is not sectorial, we consider its phase area to be $[-\infty, +\infty]$.

Then, the decentralized stability conditions in Proposition~\ref{prop:Decentralized_stability_conditions} can be interpreted as follows: 
\textit{The multi-converter system is stable if at any frequency, \textbf{either} all the gains of the converters are smaller than the gain of the power network, \textbf{or} all the phase areas of the converters are contained in the phase area of the power network while the union of all the converters' phase areas has a width less than $180^\circ$. }


The following result further considers the special case of identical R/X ratio in the network (in line with Example~\ref{example1}).

\begin{corollary}[Power grid strength]
\label{prop:Power_grid_strength}
Consider the scenario where each network line consists of a resistor and an inductor and has the same R/X ratio, corresponding to Example~\ref{example1}, such that the network dynamics can be represented by \eqref{eq:Ygrid_RX}. Then, the multi-converter system in Fig.~\ref{Fig_closed_loop} is stable if the open-loop systems are stable and for each $\omega\in [0,\infty]$, \textbf{either}
\begin{enumerate}[i)]
    \item the decentralized \textbf{gain} condition, i.e.,
    \begin{equation} \label{eq:decentralized_gain_gSCR}
    \hspace{-3mm}
    \max\limits_{i } \; \overline \sigma(\widetilde{\bf Y}_{{\rm C},i}(j\omega)) <  {\rm gSCR} \quad {\rm \textit{holds, \textbf{or}}}
    \end{equation}
    \item the decentralized \textbf{phase} condition, i.e.,
    \begin{equation}\hspace{-3mm} 
    \begin{cases}
    \textit{a) } \max\limits_{i } \; \overline \phi(\widetilde{\bf Y}_{{\rm C},i}(j\omega)) < \pi, \\
    \textit{b) }\min\limits_{i } \; \underline \phi(\widetilde{\bf Y}_{{\rm C},i}(j\omega)) > -\pi, \\
         \textit{and c) } \max\limits_{i } \; \overline \phi(\widetilde{\bf Y}_{{\rm C},i}(j\omega)) \! - \min\limits_{i } \; \underline \phi(\widetilde{\bf Y}_{{\rm C},i}(j\omega)) \!<\! \pi ,
    \end{cases} \label{eq:decentralized_phase_gSCR}
    \end{equation}
    holds, where ${\rm gSCR} = \lambda_1( {\bf S}^{-1} {\bf B}_{\rm r} ) = \underline \sigma( {\bf S}^{-\frac{1}{2}} {\bf B}_{\rm r} {\bf S}^{-\frac{1}{2}} )$ is the so-called generalized short-circuit ratio in power systems, and $\widetilde{\bf Y}_{{\rm C},i}(s) = {\bf Y}_{{\rm C},i}(s){F}_\epsilon^{-1}(s)$, $i \in \{1,..., N\}$.
\end{enumerate}
\end{corollary}
\begin{proof}
    See Appendix~\ref{Append_Proof2}.
\end{proof}

{\tr Note that Corollary~\ref{prop:Power_grid_strength} is a special case of Proposition~\ref{prop:Decentralized_stability_conditions} by considering homogeneous network dynamics (i.e., with an identical R/X ratio), where the gain of the power network becomes a constant (i.e., gSCR) and the phase area of the power network is $[-180^\circ, 180^\circ]$ at any frequency.}
In what follows, we give examples to illustrate how the decentralized stability conditions can be used to evaluate the system stability.

\begin{figure}[!t]
	\centering
	\includegraphics[width=2.8in]{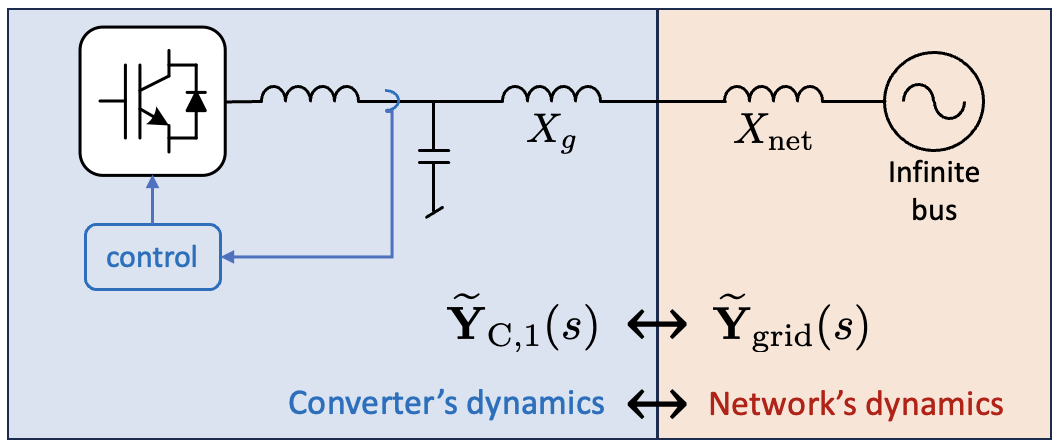}
	\vspace{-3mm}
	\caption{A single converter connected to an infinite bus.}
	\vspace{-2mm}
	\label{Fig_single_converter}
\end{figure}

\begin{example}[A single-converter system]
\label{ex:single_C}
As a starting point, let's consider a single converter connected to an infinite bus, as shown in Fig.~\ref{Fig_single_converter}. The above decentralized stability conditions evaluate the system stability by comparing the converter's dynamics with the network dynamics, as illustrated in Fig.~\ref{Fig_single_converter}, where the converter's dynamics are represented by $\widetilde{\bf Y}_{{\rm C},1}(s)$, and the network dynamics are represented by $\widetilde{\bf Y}_{\rm grid}(s)$. 
We assume that $\widetilde{\bf Y}_{{\rm C},1}(s)$ is stable, which is generally true because a converter should be designed to be stable when it is directly connected to an infinite bus (without $X_{\rm net}$). Moreover, $\widetilde{\bf Y}^{-1}_{\rm grid}(s)$ is stable because the network is passive. In a single-converter system, we simply choose ${\bf D} = 1$ and ${\bf S} = 1$. We are now ready to compare $\widetilde{\bf Y}_{{\rm C},1}(s)$ with $\widetilde{\bf Y}_{\rm grid}(s)$.

\begin{figure}[!t]
	\centering
	\includegraphics[width=3.4in]{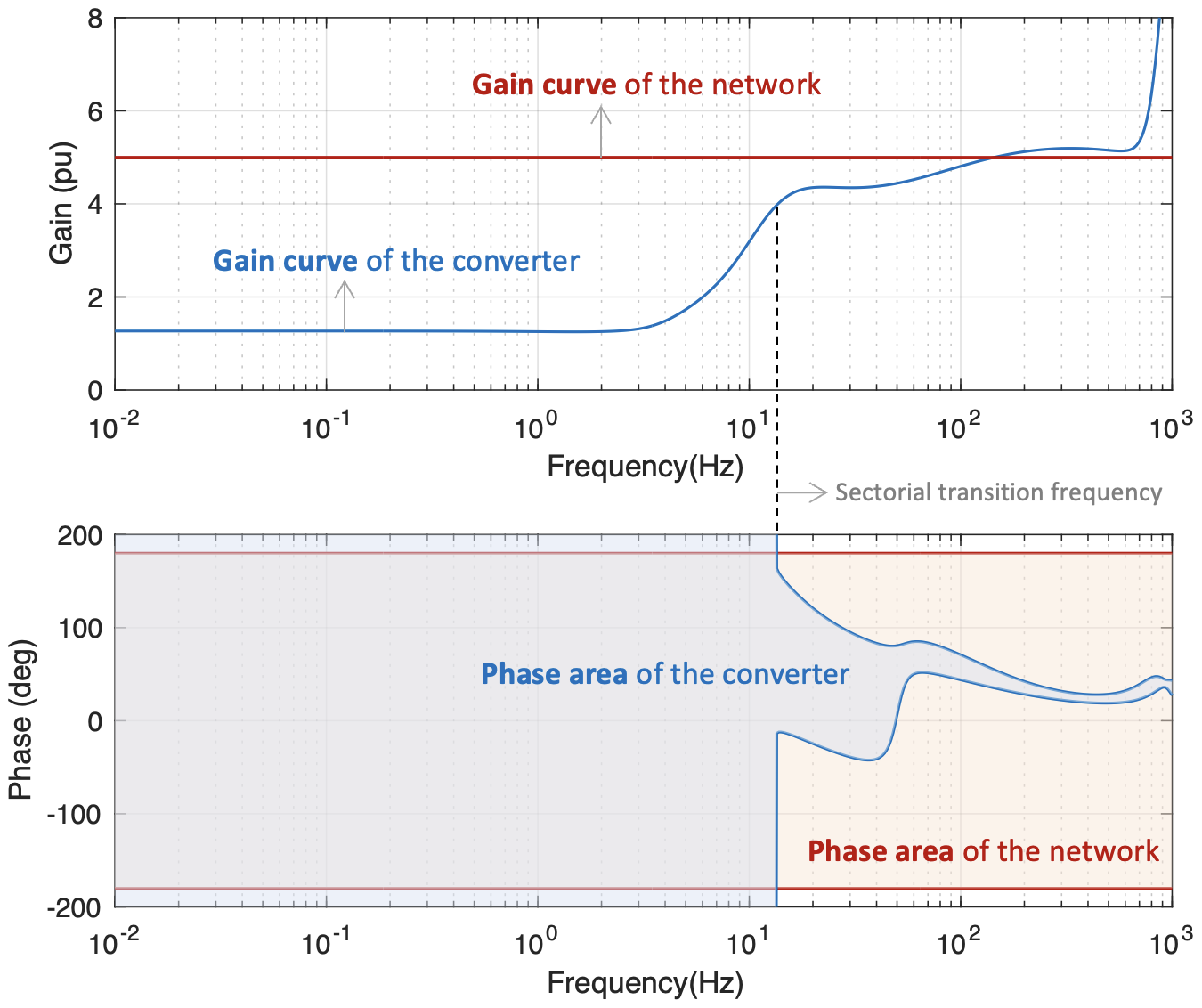}
	\vspace{-3mm}
	\caption{Gain curves and phase areas of the single-converter system (with $X_{\rm net} = 0.2~{\rm pu}$). The system is stable because 1) above the sectorial transition frequency, the phase area of the converter is contained in the phase area of the network, and 2) below the sectorial transition frequency, the gain curve of the converter is lower than the gain curve of the network.}
	\vspace{-4mm}
	\label{Fig_gain_phase_single1}
\end{figure}

Fig.~\ref{Fig_gain_phase_single1} plots the gain curve and phase area of the converter (obtained from $\widetilde{\bf Y}_{{\rm C},1}(s)$) and the gain curve and phase area of the network (obtained from $\widetilde{\bf Y}_{\rm grid}(s)$). The converter employs the GFL control in Fig.~\ref{Fig_control_diagram} to regulate active/reactive power~\cite{huang2019grid}, and the main parameters are given in Appendix~\ref{Appenxix:GFL_parameters}.

Since the network part has only an inductor, Corollary~\ref{prop:Power_grid_strength} can be used. In this case, the gain of the network is a constant (i.e., ${\rm gSCR} = 1/X_{\rm net} = 5$) and the phase area of the network is $[-180^\circ,+180^\circ]$ at any frequency. It can be seen that the phase area of the converter is contained in the phase area of the network above the sectorial transition frequency and its width is less $180^\circ$, thereby satisfying condition ii) of Corollary~\ref{prop:Power_grid_strength}; below the sectorial transition frequency, the gain curve of the converter is below the gain curve of the network,  satisfying condition i) of Corollary~\ref{prop:Power_grid_strength}. Hence, we conclude that the closed-loop system is stable, which is true, as will be validated in the simulation section.
Moreover, it can be seen that the gain of the converter at the sectorial transition frequency is 4. Hence, in this case, the system is stable if the gain of the network (i.e., gSCR) is larger than 4. Note that in this case, we include the leakage inductance of the step-up transformers in the converter side ($X_g = 0.15~{\rm pu}$), and the converter's dynamics can be considered as split from the network's dynamics at the high-voltage side (e.g., 220~kV).
\end{example}

\begin{example}[A three-converter system]
\label{ex:three_C}

\begin{figure}[!t]
	\centering
	\includegraphics[width=3.3in]{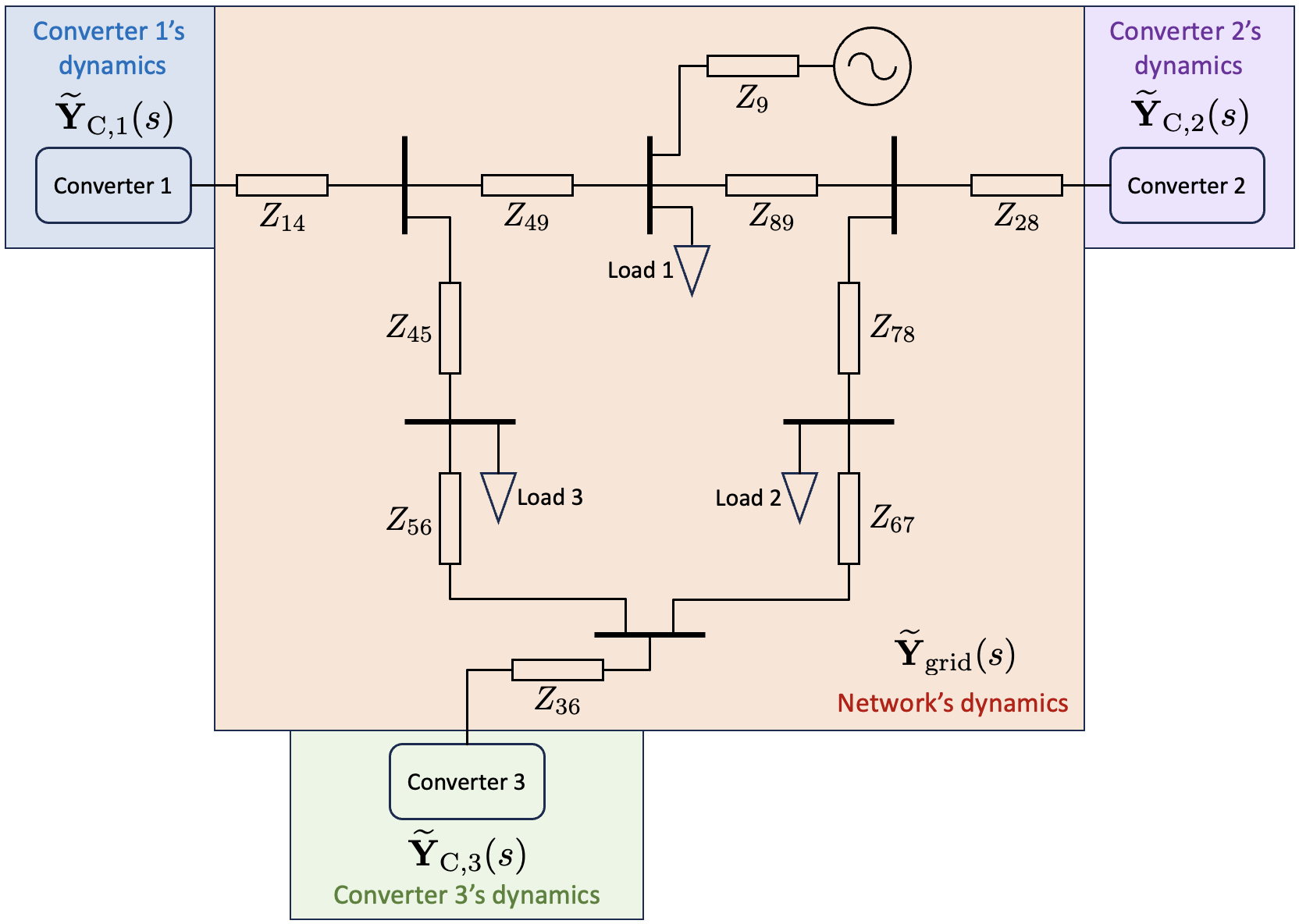}
	\vspace{-3mm}
	\caption{A three-converter test system that is split into four modules (a network and three converters).}
	\vspace{-1mm}
	\label{Fig_3_converter}
\end{figure}

\begin{figure}[!t]
	\centering
	\includegraphics[width=3.4in]{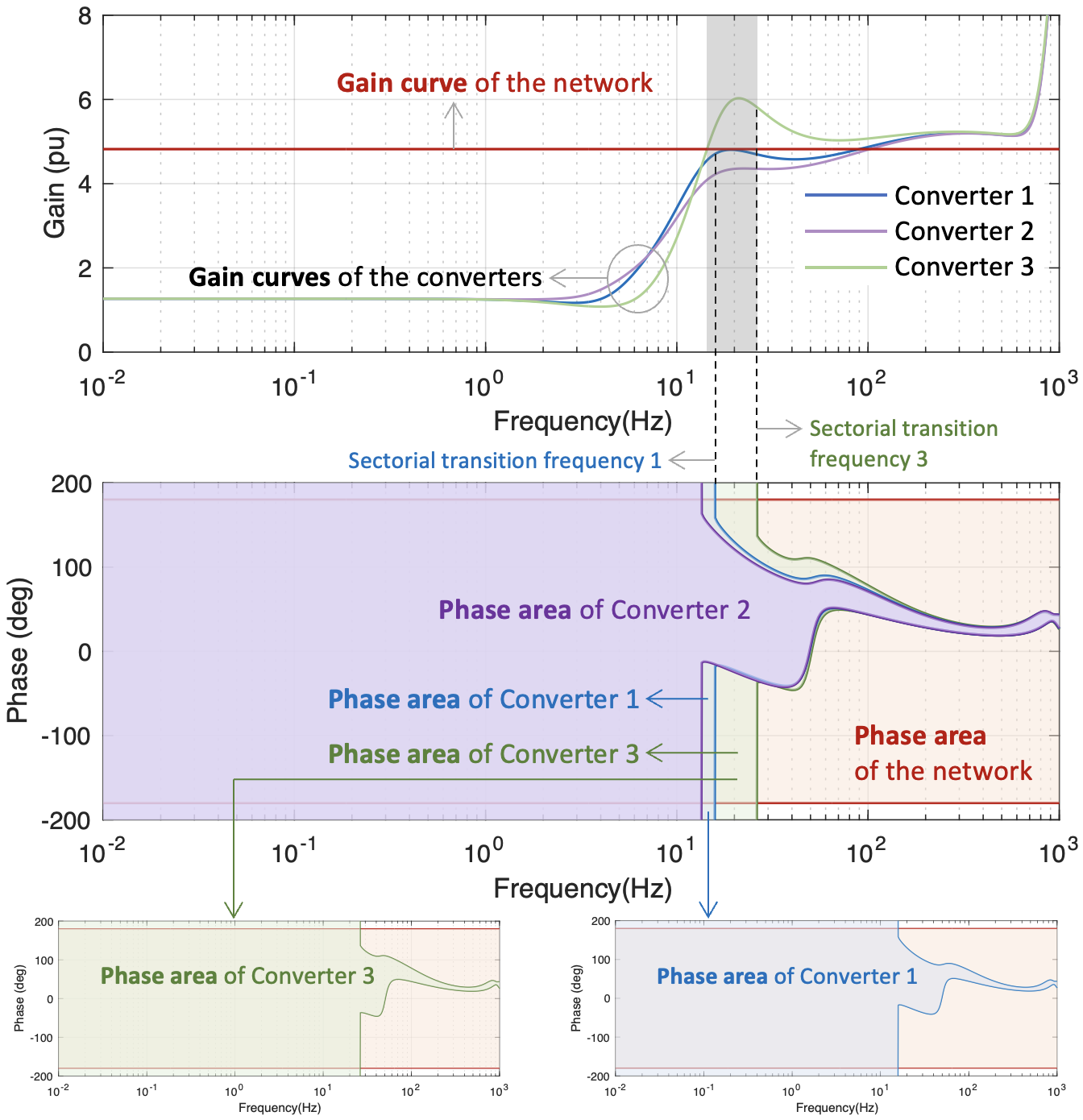}
	\vspace{-3mm}
	\caption{Gain curves and phase areas of the network and the converters.}
	\vspace{-3mm}
	\label{Fig_gain_phase_converter3_1}
\end{figure}

We consider now the three-converter system in Fig.~\ref{Fig_3_converter}, where all the converters employ GFL control. The PLL bandwidths of the three converters are respectively 70~rad/s, 40~rad/s, and 150~rad/s, and the other parameters are given in Appendix~\ref{Appenxix:GFL_parameters}. Usually, to evaluate the stability of such a multi-converter system, one needs to build the closed-loop model (e.g., the state-space model), and then compute the eigenvalues. However, this can result in complicated modeling processes and high-dimensional models. Our approach provides a simpler way by focusing on the modular converters' dynamics and network dynamics. 

Fig.~\ref{Fig_gain_phase_converter3_1} plots the gains curves and phase areas of the three converters as well as the gain curve and phase area of the power network. We choose ${\bf D} = I$ and ${\bf S} = I$ as the three converters have the same capacity. We assume an identical R/X ratio for the power network and use Corollary~\ref{prop:Power_grid_strength}. Under this setting, it can be seen that the gain of the network is a constant (${\rm gSCR} = 4.82$) and the phase area of the network is $[-180^\circ,+180^\circ]$ at any frequency. 
Converter~3 has the highest sectorial transition frequency, above which the phase areas of all the converters are contained in the phase area of the network, and moreover, the union of all the converters' phase areas has a width less than $180^\circ$. Hence, condition ii) of Corollary~\ref{prop:Power_grid_strength} is satisfied above Converter~3's sectorial transition frequency (26~Hz). However, all the converters' gain curves are below the gain curve of the network below 14.4~Hz, where condition i) of Corollary~\ref{prop:Power_grid_strength} is satisfied. Between 14.4~Hz and 26~Hz (the gray area in Fig.~\ref{Fig_gain_phase_converter3_1}), neither condition i) nor ii) of Corollary~\ref{prop:Power_grid_strength} is satisfied, and one cannot conclude that the closed-loop system is stable according to the decentralized stability conditions. As will be validated in the simulation section, the three-converter system is unstable in this case with the oscillation frequency being 15~Hz, lying exactly in the gray area.

It is worth mentioning that {\tr if a system is unstable, then the decentralized stability conditions will definitely be violated, and our approach can also tell which device in the system causes instabilities.} For instance, it can be deduced from Fig.~\ref{Fig_gain_phase_converter3_1} that Converter~3 results in the gray area where Corollary~\ref{prop:Power_grid_strength} is not satisfied. To be specific, if Converter~3 employs a better design, for instance, having the same design as Converter~1 (with the PLL bandwidth being 40~rad/s), then condition ii) of Corollary~\ref{prop:Power_grid_strength} is satisfied above Converter~1's sectorial transition frequency (15.8~Hz). Moreover, below this frequency, all the converters' gain curves are below the gain curve of the network, thereby satisfying condition i) of Corollary~\ref{prop:Power_grid_strength}. Hence, by comparison, we conclude that the closed-loop system would have been stable if Converter~3 has a better design, and it is the high PLL bandwidth of Converter~3 (150~rad/s) that causes the instability. 


\end{example}

\vspace{-2mm}

\section{Conditions for {\tr SG-}GFM-GFL Hybrid Systems}

{\tr The above analysis focuses on the interaction between multiple GFL converters and the power network. However, future power systems will contain GFL converters, GFM converters, and synchronous generators (SGs) at the same time. In what follows, we further consider GFM converters and SGs in our approach to make it suitable for analyzing future bulk power systems. We start with systems that contain both GFL and GFM converters.}

\subsection{\tr System partitioning and equivalent subsystems}

In Example~\ref{ex:three_C}, the system contains multiple GFL converters, and the analysis indicates that a lower gain of a converter is better for the closed-loop stability, which requires the converter to have a lower equivalent admittance. However, a GFM converter may have high admittance due to its voltage source behavior~\cite{xin2022many} and could be difficult to satisfy the gain condition. To this end, we propose to merge the GFM converters into an equivalent network, as illustrated in Fig.~\ref{Fig_equivalent_syst} and will be elaborated below.
Note that there is no unique way to split the system (into devices' dynamics and network's dynamics, as shown in Fig.~\ref{Fig_closed_loop}), and different ways may result in different levels of conservatism. We next present a new way to split the system dynamics, which is based on the idea of incorporating some of the converters (especially GFM converters) into the network dynamics and deriving the equivalent subsystems.

\begin{lemma}[Equivalent subsystems]
\label{lem:Equivalent_systems}
The multi-device system in Fig.~\ref{Fig_closed_loop} is stable if the following two subsystems are stable:
\vspace{0mm}
\begin{subequations}
\begin{empheq}[left={\empheqlbrace}]{align}
& {\rm subsystem~1:} \begin{bmatrix} 0 & \\ & {\bf Y}_{{\rm C}\{2\}}(s) \end{bmatrix} \; \# \; {\bf Y}^{-1}_{\rm grid}(s), \label{eq:subsys1} \\
& {\rm subsystem~2:} \; {\bf Y}_{{\rm C}\{1\}}(s) \; \# \; {\bf Y}^{-1}_{\rm gridC}(s), \label{eq:subsys2}
\end{empheq}
\end{subequations}
\vspace{0mm}
where 
\begin{equation}\label{eq:equil_network}
{\bf Y}_{\rm gridC}(s) = {\bf Y}_{\rm g}^1(s) - {\bf Y}_{\rm g}^2(s) [{\bf Y}_{{\rm C}\{2\}}(s) + {\bf Y}_{\rm g}^4(s)]^{-1} {\bf Y}_{\rm g}^3(s),
\end{equation}
\begin{equation}
    \begin{bmatrix}
    {\bf Y}_{\rm g}^1(s) & {\bf Y}_{\rm g}^2(s) \vspace{1mm} \\ 
    {\bf Y}_{\rm g}^3(s) & {\bf Y}_{\rm g}^4(s) \\
\end{bmatrix} := {\bf Y}_{\rm grid}(s), \; {\rm and}
\end{equation}
\begin{equation*}\label{eq:two_sets}
\begin{bmatrix} {\bf Y}_{{\rm C}\{1\}}(s) & \\ & {\bf Y}_{{\rm C}\{2\}}(s) \end{bmatrix} := ({\bf S}\otimes I_2)e^{J {\vartheta}} \; {\bf Y}_{\rm C}^N(s) \; e^{-J {\vartheta}},
\end{equation*}
with the dimensions of ${\bf Y}_{{\rm C}\{1\}}(s)$ and ${\bf Y}_{{\rm C}\{2\}}(s)$ match the dimensions of ${\bf Y}_{\rm g}^1(s)$ and ${\bf Y}_{\rm g}^4(s)$, respectively. Moreover, if subsystem~1 is stable, then ${\bf Y}^{-1}_{\rm gridC}(s)$ is stable. This is important to ensure that the open-loop systems of subsystem~2 is stable so the mixed small gain-phase theorem can be applied.
\end{lemma}
\begin{proof}
    See Appendix~\ref{Append_Proof3}.
\end{proof}

It is worth mentioning that~\eqref{eq:char_poly_equil} in the proof indicates that the closed-loop poles of the original multi-converter system are also closed-loop poles of subsystem~1 in~\eqref{eq:subsys1} or subsystem~2 in~\eqref{eq:subsys2}. Hence, instead of directly studying the original system, we can focus on the stability of the two subsystems. 
Notice that Lemma~\ref{lem:Equivalent_systems} divides all the converters into two sets (i.e., ${\bf Y}_{{\rm C}\{1\}}(s)$ and ${\bf Y}_{{\rm C}\{2\}}(s)$), as shown in~\eqref{eq:two_sets}. In our setting, the GFL converters are included in ${\bf Y}_{{\rm C}\{1\}}(s)$ while the GFM converters are included in ${\bf Y}_{{\rm C}\{2\}}(s)$. {\tr If subsystem~1 in~\eqref{eq:subsys1} contains only one GFM converter, namely, a single-converter-infinite-bus system, then its stability can be analyzed using standard state-space method, impedance-based method, or by using our approach in this paper (see Example~\ref{ex:single_C}). If subsystem~1 contains multiple GFM converters, then its stability can also be analyzed using the proposed decentralized stability conditions (see Example~\ref{ex:68_bus} below).}
After analyzing the stability of subsystem~1, we can focus on the stability of subsystem~2 in~\eqref{eq:subsys2} where the \textit{\tr equivalent} network contains GFM dynamics {\tr (as it can be seen from~\eqref{eq:equil_network} that ${\bf Y}_{\rm gridC}(s)$ is computed based on ${\bf Y}_{{\rm C}\{2\}}(s)$)}. The example below shows how Lemma~\ref{lem:Equivalent_systems} enables the analysis of a GFM-GFL hybrid system, as also illustrated in Fig.~\ref{Fig_equivalent_syst}.

\begin{example}[A three-converter system where one converter employs GFM control]
\label{ex:three_C_GFM}

Consider again the three-converter system in Fig.~\ref{Fig_3_converter}, where Converters~1 and~2 employ GFL control while Converter~3 employs GFM control. {\tr The GFM control scheme and parameters are provided in Appendix~\ref{Appenxix:GFL_parameters}}. We firstly partition the system into two subsystems according to Lemma~\ref{lem:Equivalent_systems}, as shown in Fig.~\ref{Fig_equivalent_syst}. It can be seen that subsystem~1 is simply a GFM converter connected to an infinite bus, which we assume to be stable. Hence, we can focus on subsystem~2 to analyze the stability of the original system. In subsystem~2, two GFL converters are connected to an equivalent network whose dynamics is described by~\eqref{eq:equil_network}. Then, we can use Proposition~\ref{prop:Decentralized_stability_conditions} to analyze the stability of subsystem~2, by comparing the converters' dynamics ($\widetilde{\bf Y}_{{\rm C},1}(s)$ and $\widetilde{\bf Y}_{{\rm C},2}(s)$) with the equivalent network's dynamics, i.e.,
\begin{equation}\label{eq:equil_network_1}
    \widetilde{\bf Y}_{\rm gridC}(s) := ({\bf S}^{-\frac{1}{2}} \otimes I_2) {\bf Y}_{\rm gridC}(s) ( {\bf S}^{-\frac{1}{2}} {\bf D} \otimes {\widetilde F}_\epsilon^{-1}(s)),
\end{equation}
which is obtained from~\eqref{eq:equil_network} and~\eqref{eq:Y_C_grid}.

\begin{figure}[!t]
	\centering
	\includegraphics[width=3.4in]{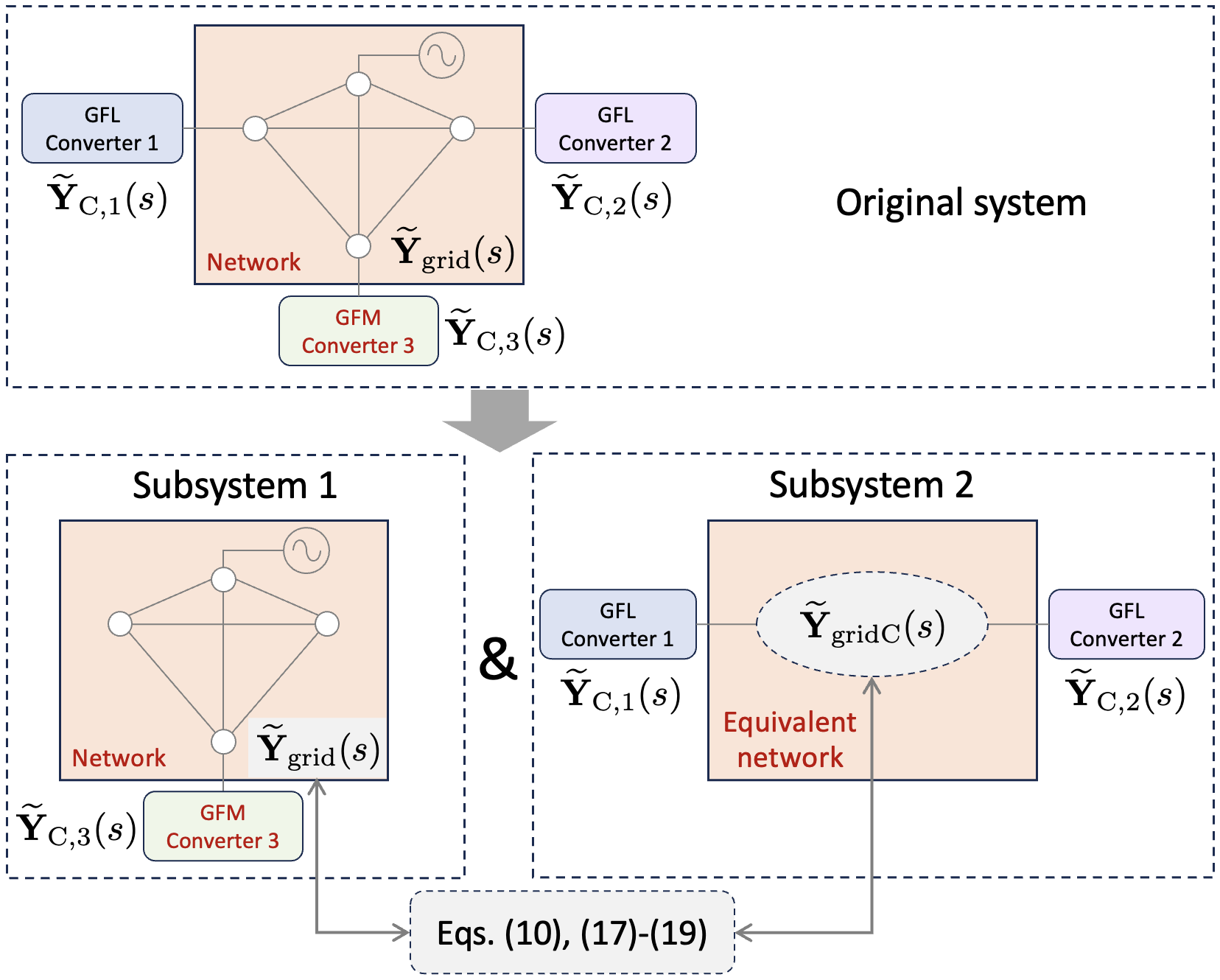}
	\vspace{-3mm}
	\caption{Equivalent subsystems of the three-converter system. {\tr The network in subsystem~1 and the equivalent network in subsystem~2 are related via Eqs.~\eqref{eq:Y_C_grid},~\eqref{eq:equil_network}-\eqref{eq:equil_network_1}.}}
	\vspace{0mm}
	\label{Fig_equivalent_syst}
\end{figure}

\begin{figure}[!t]
	\centering
	\includegraphics[width=3.4in]{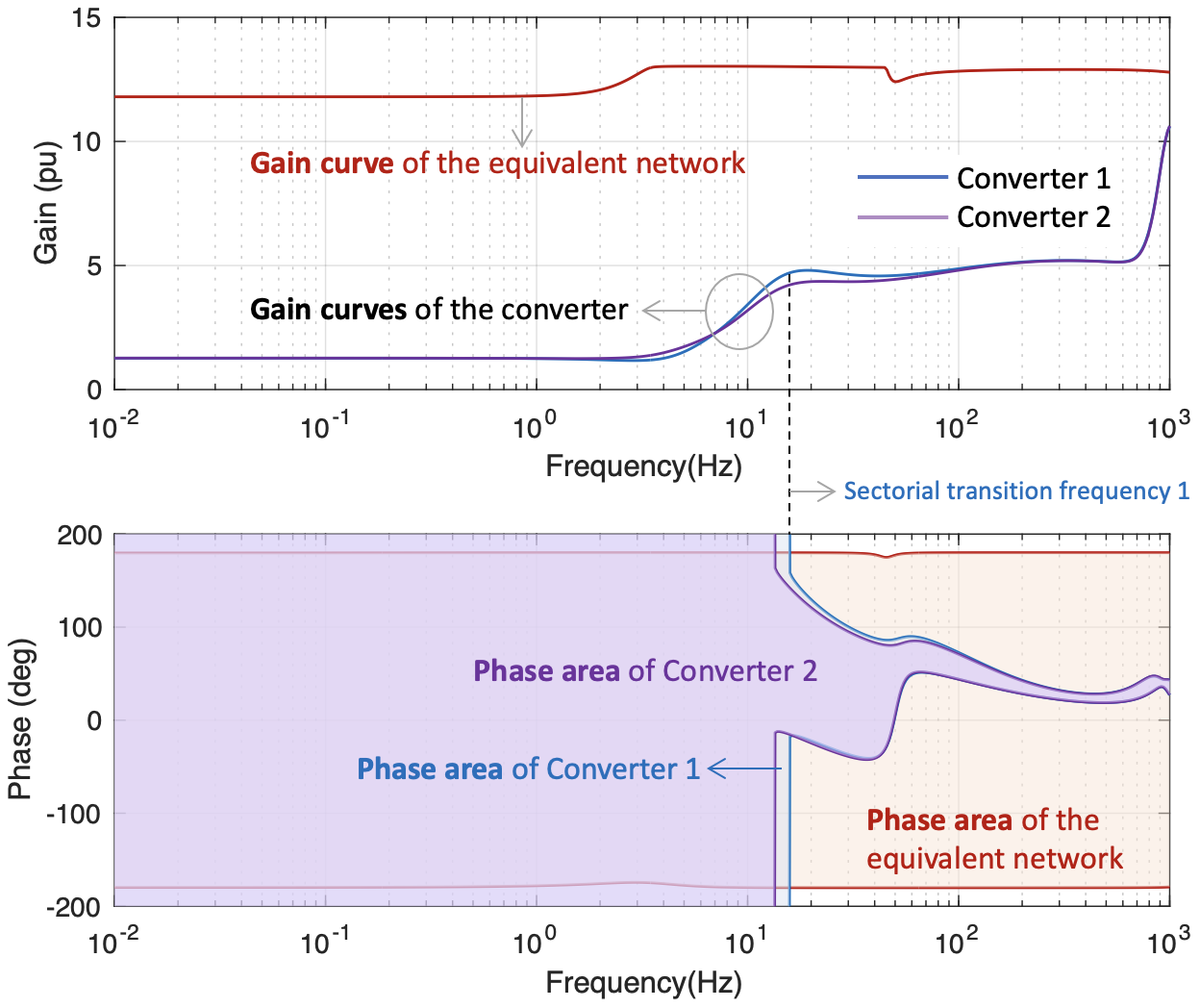}
	\vspace{-4mm}
	\caption{Gain curves and phase areas of subsystem~2.}
	\vspace{-3mm}
	\label{Fig_gain_phase_GFM}
\end{figure}

Fig.~\ref{Fig_gain_phase_GFM} plots the gain curves and phase areas of Converters~1 and~2 as well as the gain curve and phase area of the equivalent network (computed from $\widetilde{\bf Y}_{\rm gridC}(s)$). It shows that the gain curve and phase area of the equivalent network are not constant across all the frequencies, because the GFM dynamics are included in the equivalent network (see~\eqref{eq:equil_network}). Under the same network parameters, we can see by comparing Fig.~\ref{Fig_gain_phase_GFM} and Fig.~\ref{Fig_gain_phase_converter3_1} that incorporating the GFM dynamics into the network part increases the gain curve of the equivalent network, and therefore improves the system stability. This is consistent with existing works showing that the installation of GFM converters improves the small-signal stability of the system~\cite{xin2022many,yang2020placing}. In Fig.~\ref{Fig_gain_phase_GFM}, above the sectorial transition frequency of Converter~1, the phase areas of the two converters are contained in the phase area of the equivalent network, and the union of the two converters' phase areas has a width less than $180^\circ$. Moreover, below the sectorial transition frequency, the gain curves of the two converters are below the gain curve of the equivalent network. Hence, we conclude that the system is stable according to Proposition~\ref{prop:Decentralized_stability_conditions}.
\end{example}

{\tr
\subsection{Extension to SG-GFM-GFL hybrid systems}

The above example studies a GFM-GFL hybrid system, where the effects of SGs are implicitly modeled into the infinite bus for simplicity, corresponding to the cases where the capacities of converters are relatively small. However, with the increasing penetration of renewable energy generators, the effects of SGs should be taken into account as GFM and GFL converters will dynamically interact with SGs in future power systems. Hence, in what follows, we further consider SGs in our approach.

Notice that the previous gain-phase analysis is based on the admittance models of devices (e.g., GFL converters, GFM converters), as shown in Fig.~\ref{Fig_closed_loop} where ${\bf Y}_{{\rm C},i}(s)$ denotes the admittance matrix of the $i$-th device. The admittance matrix of a device is essentially the $2 \times 2$ transfer function matrix that describes how the device's $dq$ current signals respond to the perturbations in the $dq$ voltage signals, namely, the transfer function matrix from voltage inputs to current outputs. Therefore, the dynamics of an SG can also be captured by deriving its admittance model~\cite{huang2020damping}, and we can still use the closed-loop diagram in Fig.~\ref{Fig_closed_loop} to represent the dynamics of an SG-GFM-GFL hybrid system. To be specific, if the $i$-th device is an SG, then ${\bf Y}_{{\rm C},i}(s)$ should be set as the SG's admittance matrix. Moreover, since SGs also have GFM ability, we group them with GFM converters into ${\bf Y}_{{\rm C}\{2\}}(s)$ in~\eqref{eq:subsys1} (i.e., in subsystem~1).

In summary, when dealing with a system that contains SGs, GFM converters, and GFL converters, we group the SGs and GFM converters into ${\bf Y}_{{\rm C}\{2\}}(s)$ and group the GFL converters into ${\bf Y}_{{\rm C}\{1\}}(s)$. According to Lemma~\ref{lem:Equivalent_systems}, the resulting subsystems~1 describes how the SGs and the GFM converters are coupled through the power network, and the resulting subsystems~2 describes how the GFL converters are coupled through an \textit{equivalent} network (see~\eqref{eq:equil_network}-\eqref{eq:equil_network_1}). It can be seen from~\eqref{eq:equil_network} that ${\bf Y}_{\rm gridC}(s)$ is obtained from both ${\bf Y}_{\rm grid}(s)$ and ${\bf Y}_{{\rm C}\{2\}}(s)$, which can be interpreted as merging the GFM dynamics into the network. In Example~\ref{ex:three_C_GFM}, we showed that merging the GFM dynamics into the network is equivalent to increasing the gain of the network, which justifies the advantage of installing GFM converters. Note that in this paper, we need both ${\bf Y}_{\rm grid}(s)$ and ${\bf Y}_{{\rm C}\{2\}}(s)$ to compute ${\bf Y}_{\rm gridC}(s)$ and therefore $\widetilde{\bf Y}_{\rm gridC}(s)$ in~\eqref{eq:equil_network_1}, while we will investigate how to reduce the requirement of ${\bf Y}_{{\rm C}\{2\}}(s)$ in future works (e.g., by using reduced-order models or by approximating GFM devices by voltage sources).

Since subsystem~1 and subsystem~2 may both contain heterogeneous devices, we sequentially use Proposition~\ref{prop:Decentralized_stability_conditions} (or Corollary~\ref{prop:Power_grid_strength} if applicable) to evaluate their stability. According to Lemma~\ref{lem:Equivalent_systems}, the original SG-GFM-GFL hybrid system is stable if both subsystem~1 and subsystem~2 are stable. We provide an example (on a modified 68-bus system) below to illustrate the effectiveness of our approach on SG-GFM-GFL hybrid systems.
}

{\tr

\begin{example}[Stability of a modified 68-bus system integrated with 2 SGs, 4 GFM converters, and 9 GFL converters]
\label{ex:68_bus}

\tr

\begin{figure}[!t]
	\centering
	\includegraphics[width=3.4in]{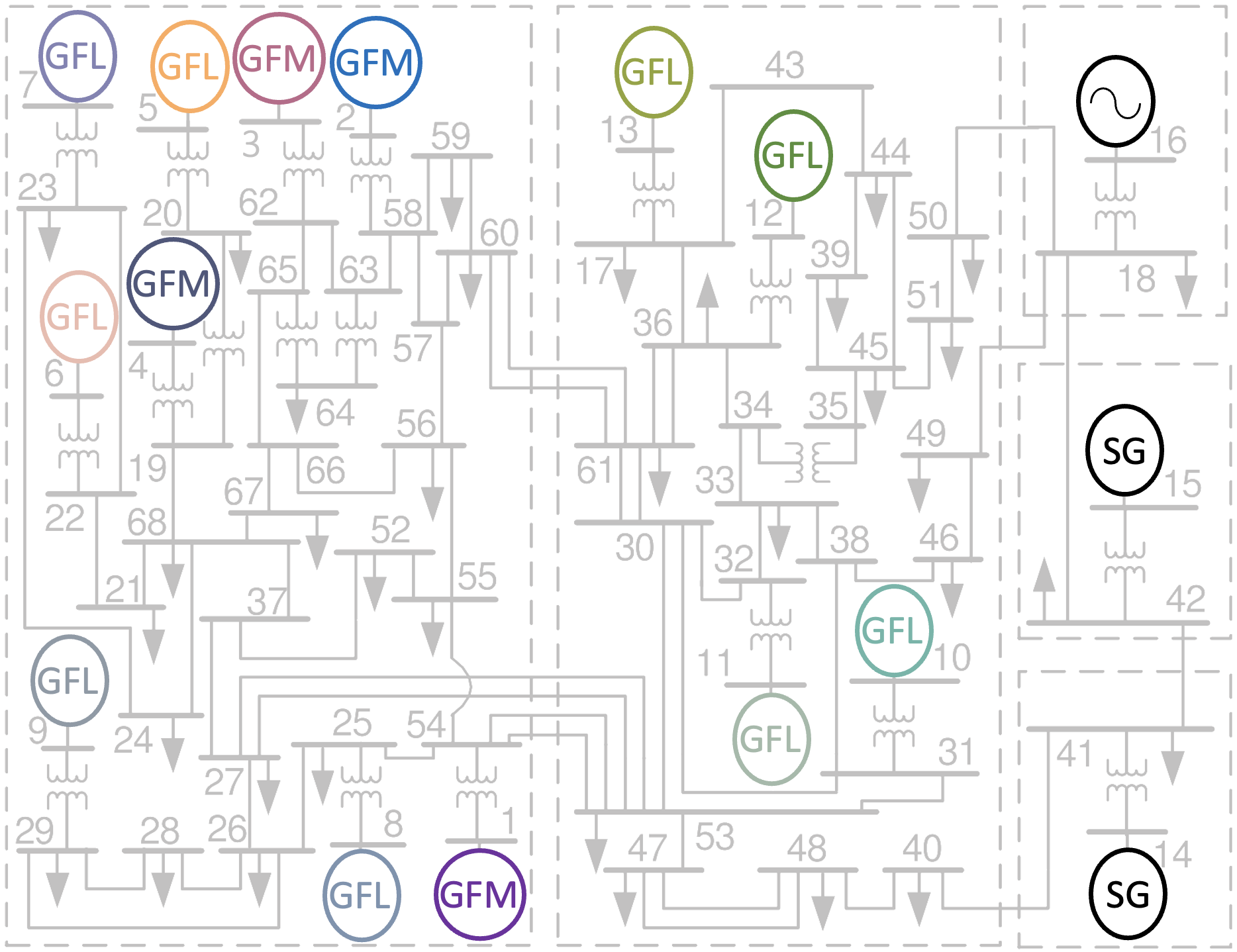}
	\vspace{-3mm}
	\caption{\tr A modified 68-bus system integrated with 2 SGs, 4 GFM converters, and 9 GFL converters. (Note that the original setting of the IEEE 68-bus system can be found in~\cite{canizares2016benchmark}.)}
	\vspace{0mm}
	\label{Fig_68_bus}
\end{figure}

\tr

Consider a modified IEEE 68-bus system as shown in Fig.~\ref{Fig_68_bus}, where four GFM converters are connected to Buses~1~$\sim$~4 (denoted by GFM converters~1~$\sim$~4), nine GFL converters are connected to Buses~5~$\sim$~13 (denoted by GFL converters~5~$\sim$~13), and two SGs are connected to Buses~14 and 15 (denoted by SGs~14 and 15). Bus~16 of the systems is modeled as an infinite bus to emulate a remote (weakly-connected) area. The topology of the power network and load profiles are the same as the original 68-bus systems in~\cite{canizares2016benchmark}. The Pi-model is used for the transmission lines, and we allow different R/X ratios for different transmission lines, aligned with practical settings. The detailed parameters of this modified IEEE 68-bus system can be found in Appendix~\ref{Appenxix:GFL_parameters}. Note that this system has heterogeneous R/X ratios, and we choose $\widetilde \epsilon = 0.1$ in ${\widetilde F}_\epsilon(s)$ when splitting the system dynamics, which is the R/X ratio of all the converters' grid-side impedance. Moreover, we further consider voltage normalization and rated frequency feed-forward in the PLL design in the GFL converters~\cite{golestan2016three}.

Note that it is usually difficult to analyze the stability of such a complex system, due to its high dimension and the complicated coupling among heterogeneous devices (SGs, GFM and GFL converters). Nonetheless, our approach can be applied to gain insights into how the devices interact with the network and how to improve the overall system stability via device-level control design. Moreover, as will be shown below, our approach can be used to find out which devices cause instabilities and which devices should be re-designed.

\tr

Following the previous discussions in this subsection, the first step is to partition the system into two subsystems according to Lemma~\ref{lem:Equivalent_systems}, as shown in Fig.~\ref{Fig_68_equivalent}. It can be seen that the four GFM converters and the two SGs are connected to the original network in subsystem~1, while the nine GFL converters are connected to the equivalent network in subsystem~2. 

\tr

\begin{figure}[!t]
	\centering
	\includegraphics[width=3.4in]{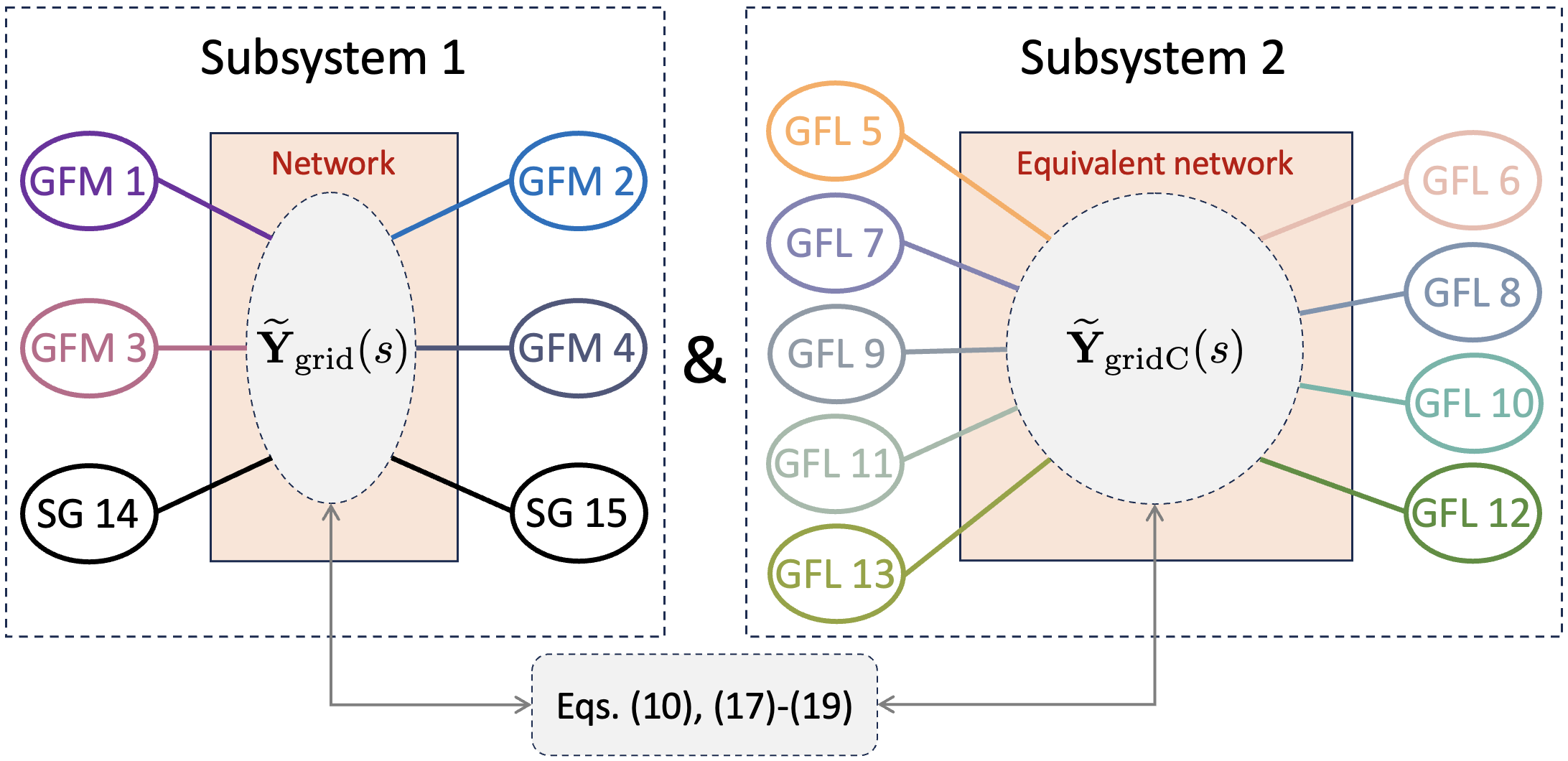}
	\vspace{-3mm}
	\caption{\tr Equivalent subsystems of the modified 68-bus system. The network in subsystem~1 and the equivalent network in subsystem~2 are related via Eqs.~\eqref{eq:Y_C_grid},~\eqref{eq:equil_network}-\eqref{eq:equil_network_1}.}
	\vspace{0mm}
	\label{Fig_68_equivalent}
\end{figure}

We next analyze the stability of subsystem~1. Fig.~\ref{Fig_68_GFM_SG} plots the gain curves and phase areas of the four GFM converters and the two SGs as well as the gain curve and phase area of
the network (the GFL converter nodes are eliminated through Kron reduction as subsystem~1 does not contain GFL converters). The GFM converters 2, 3, and 4 employ the same design and parameters so their gain curves and phase areas are the same. The two SGs have the same parameters so their gain curves and phase areas are the same. It can be seen that above 0.77~Hz, the phase areas of all the GFM converters and the SGs are contained in the phase area of the network, while below 0.77~Hz, the gain curves of all the GFM converters and the SGs are below the gain curve of the network. Hence, according to Proposition~\ref{prop:Decentralized_stability_conditions}, subsystem~1 is stable.

\tr

\begin{figure}[!t]
	\centering
	\includegraphics[width=3.4in]{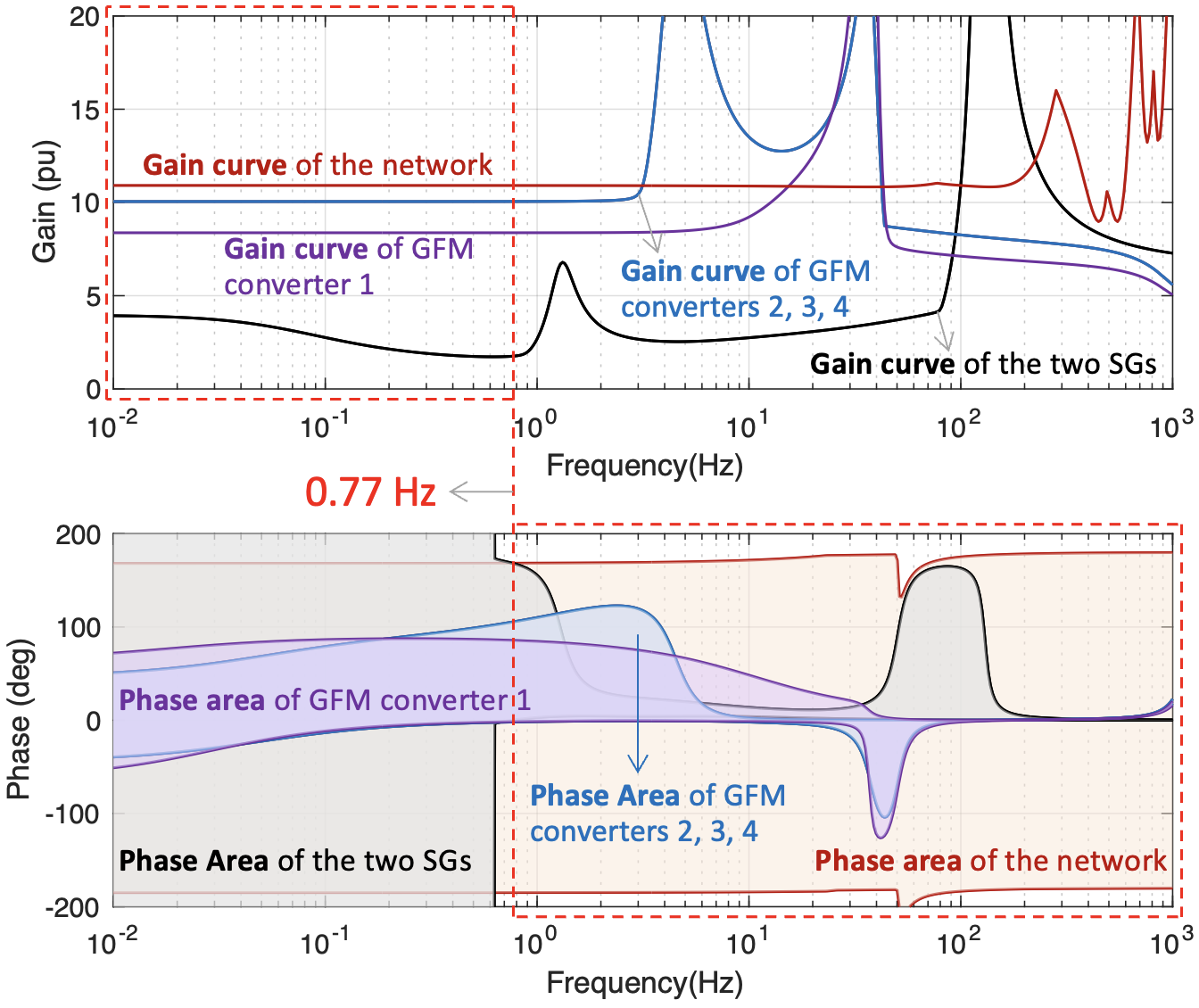}
	\vspace{-3mm}
	\caption{\tr Gain curves and phase areas of subsystem 1 of the modified 68-bus system, where the \textit{original} network is focused.}
	\vspace{0mm}
	\label{Fig_68_GFM_SG}
\end{figure}

\tr

Then, we analyze the stability of subsystem~2. To illustrate the applicability of our approach in analyzing heterogeneous dynamics, we set the PLL bandwidths of GFL converters~5~$\sim$~13 to be respectively 10~rad/s, 15~rad/s, 20~rad/s, 25~rad/s, 30~rad/s, 35~rad/s, 40~rad/s, 60~rad/s, 70~rad/s, that is, all the GFL converters have different dynamics. 

Fig.~\ref{Fig_68_unstable} plots the gain curves and phase areas of the nine GFL converters as well as the gain curve and phase area of the
equivalent network. A zoom-in version of Fig.~\ref{Fig_68_unstable} is provided in Fig.~\ref{Fig_68_unstable_zoomin}~(a). 
It can be seen that GFL converter~13 has the highest sectorial transition frequency, above which the phase areas of all the GFL converters are contained in the phase area of the equivalent network, and moreover, the union of all the converters' phase areas has a width less than $180^\circ$. Hence, condition ii) of Proposition~\ref{prop:Decentralized_stability_conditions} is satisfied above GFL converter~13's sectorial transition frequency (14.1~Hz). However, all the converters' gain curves are below the gain curve of the network below 10.2~Hz, where condition i) of Proposition~\ref{prop:Decentralized_stability_conditions} is satisfied. Between 10.2~Hz and 14.1~Hz (the gray area in Fig.~\ref{Fig_68_unstable}), neither condition i) nor ii) of Proposition~\ref{prop:Decentralized_stability_conditions} is satisfied, and one cannot conclude that the closed-loop system is stable according to the decentralized stability conditions. As will be validated in the simulation section, the modified 68-bus system is unstable in this case with the oscillation frequency being 11.3~Hz, lying exactly in the gray area.

It can be observed from Fig.~\ref{Fig_68_unstable} and Fig.~\ref{Fig_68_unstable_zoomin}~(a) that the gains of GFL converters 12 and 13 are too high, which leads to the gray area in Fig.~\ref{Fig_68_unstable}. One can further deduce that the system could have been stable if GFL converters 12 and 13 employed a better design. For instance, Fig.~\ref{Fig_68_unstable_zoomin}~(b) plots the new gain curves and phase areas of GFL converters 12 and 13 when their PLL bandwidths are respectively reduced from 60~rad/s and 70~rad/s to 5~rad/s (see the red dash lines). After reducing the PLL bandwidths, GFL converter 11 has the highest sectorial transition frequency (12.9~Hz), above which the phase areas of all the GFL converters are contained in the phase area of the equivalent network. Below this frequency, the gain curves of all the GFL converters are below the gain curve of the equivalent network. Therefore, according to Proposition~\ref{prop:Decentralized_stability_conditions}, subsystem~2 is stable. Since subsystem~1 is also stable, according to Lemma~\ref{lem:Equivalent_systems}, we conclude that the modified 68-bus system is stable in this scenario (i.e., with reduced PLL bandwidths in Fig.~\ref{Fig_68_unstable_zoomin}~(b)).

\begin{figure}[!t]
	\centering
	\includegraphics[width=3.4in]{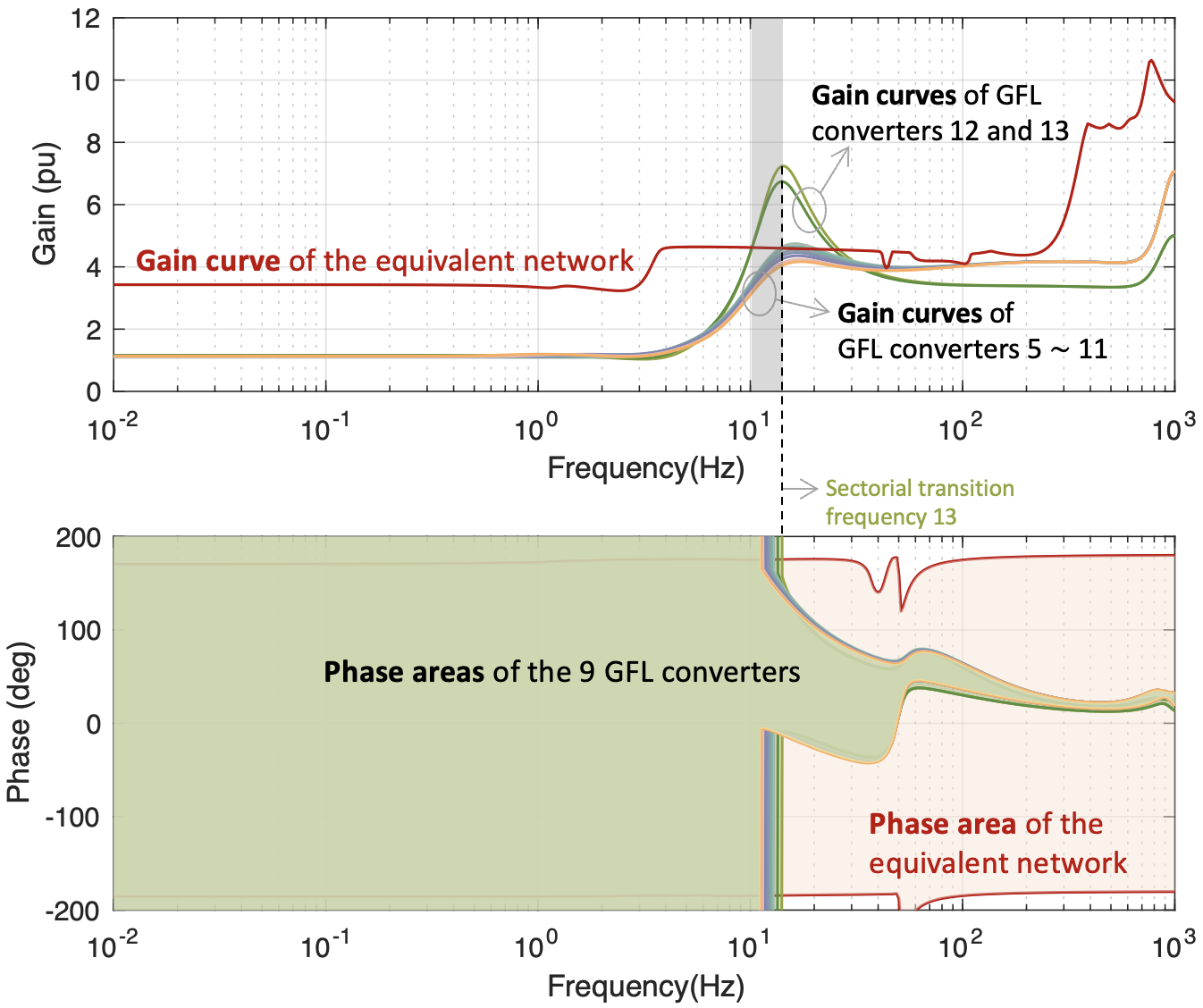}
	\vspace{-3mm}
	\caption{\tr Gain curves and phase areas of subsystem 2 of the modified 68-bus system, where the \textit{equivalent} network is focused.}
	\vspace{0mm}
	\label{Fig_68_unstable}
\end{figure}

\begin{figure}[!t]
	\centering
	\includegraphics[width=3.4in]{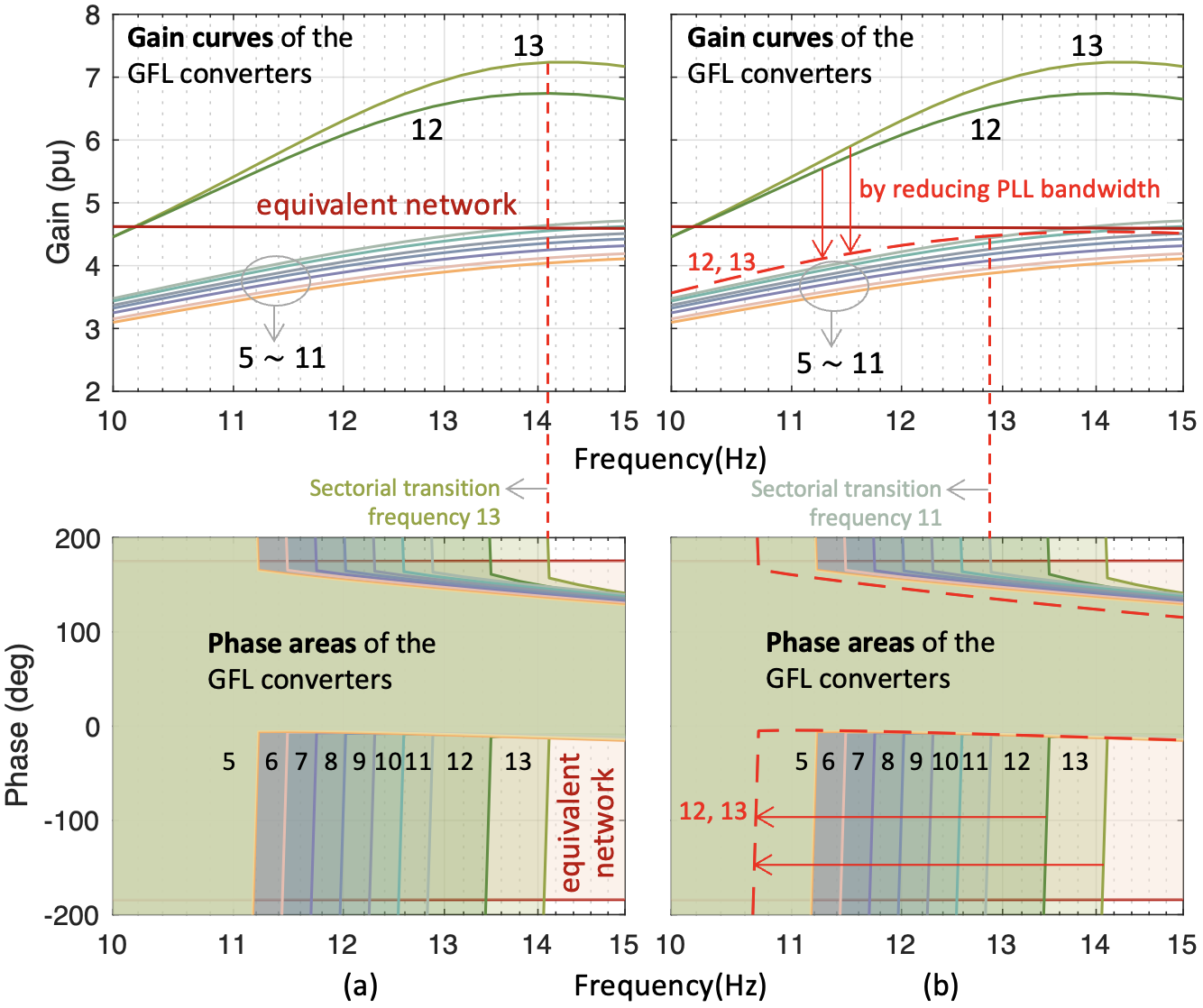}
	\vspace{-3mm}
	\caption{\tr Gain curves and phase areas of subsystem 2: (a) a zoom-in version of Fig.~\ref{Fig_68_unstable}, and (b) reducing the PLL bandwidths of GFL converters 12 and 13 to 5~rad/s (the new gain curves and phase areas of GFL converters 12 and 13 are plotted using red dash lines).}
	\vspace{0mm}
	\label{Fig_68_unstable_zoomin}
\end{figure}

\end{example}
}

        \section{Simulations results}
	
In what follows, we {\tr use high-fidelity power system models (including nonlinear dynamics) to} provide time-domain simulation results and verify the previous analysis that uses the decentralized stability conditions. 

\subsection{Time-domain simulation results of Example~\ref{ex:single_C}}

Fig.~\ref{Fig_sim_1} shows the responses of the single-converter system in Fig.~\ref{Fig_single_converter}. At $t=0.5~{\rm s}$, the network reactance $X_{\rm net}$ steps from $0.1~{\rm pu}$ to respectively $0.2~{\rm pu}$ (black line), $0.25~{\rm pu}$ (yellow line), and $0.295~{\rm pu}$ (blue line). It can be seen that the system is stable with $X_{\rm net} = 0.2~{\rm pu}$, consistent with the analysis in Example~\ref{ex:single_C}. The gain-phase analysis in Fig.~\ref{Fig_gain_phase_single1} gives a sufficient stability condition indicating that the system is stable if the gain of the network is larger than or equal to 4, i.e., $X_{\rm net} \le 0.25~{\rm pu}$, which is verified by the yellow line in Fig.~\ref{Fig_sim_1}. The blue line in Fig.~\ref{Fig_sim_1} shows that the system is marginally stable with $X_{\rm net} = 0.295~{\rm pu}$. The gap between 0.25 and 0.295 is due to the fact that our approach gives sufficient stability conditions. However, the above result shows that the conservatism is acceptable and reasonable, as the gap between the marginal value obtained from our approach (i.e., 0.25) and the true marginal value (i.e., 0.295) is only 0.045.

\begin{figure}[!t]
	\centering
	\includegraphics[width=3.1in]{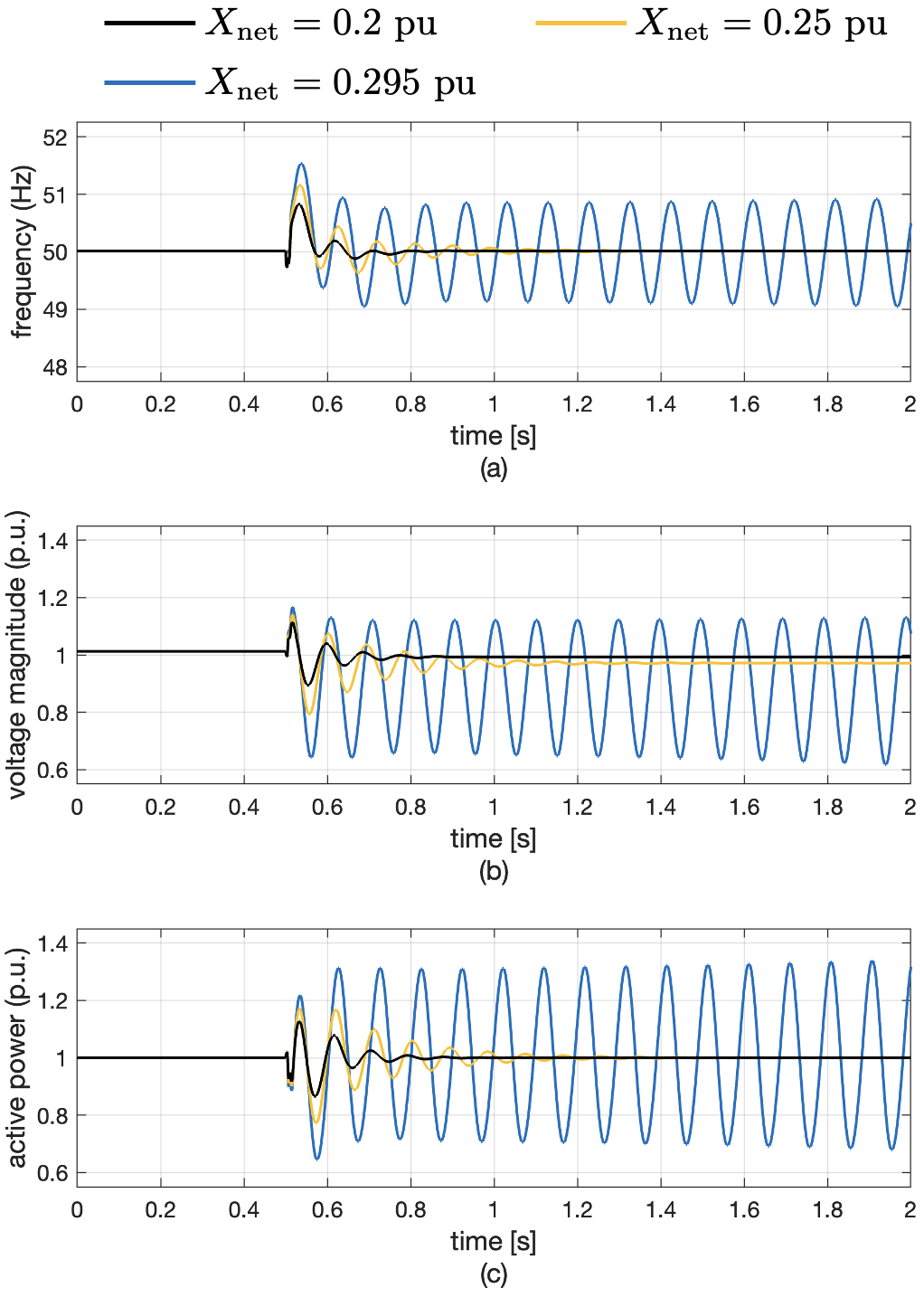}
	\vspace{-4mm}
	\caption{Time-domain responses of the single-converter system in Fig.~\ref{Fig_single_converter}: (a)~frequency, (b)~voltage magnitude, and (c)~active power.}
	\vspace{-1mm}
	\label{Fig_sim_1}
\end{figure}

\subsection{Time-domain simulation results of Examples~\ref{ex:three_C} and~\ref{ex:three_C_GFM}}

Fig.~\ref{Fig_sim_2} shows the time-domain responses of the three-converter system in Fig.~\ref{Fig_3_converter}, in which the gSCR of the system changes from 6 to 4.82 at $t=0.5~{\rm s}$ to emulate changes in the power network (e.g., a transmission line is opened).
Fig.~\ref{Fig_sim_2}~(a) corresponds to the setting in Example~\ref{ex:three_C} (i.e., all the converters use GFL control) and shows that the system is unstable (close to marginally stable) with the oscillation frequency being 15~Hz, consistent with the analysis in Fig.~\ref{Fig_gain_phase_converter3_1}. Fig.~\ref{Fig_sim_2}~(b) corresponds to the setting in Example~\ref{ex:three_C_GFM} where Converter~3 adopts GFM control, which shows that the system is stable, consistent with the analysis in Fig.~\ref{Fig_gain_phase_GFM}. By comparing Fig.~\ref{Fig_sim_2}~(a) with (b), we can deduce that the system stability is improved by changing Converter~3 from GFL control to GFM control, in line with the analysis in Example~\ref{ex:three_C_GFM} showing that operating Converter~3 in GFM mode increases the gain (or gSCR) of the equivalent network.

{\tr
We further test the performance of the three-converter system when it is not at a steady state due to, for instance, the fluctuation of renewable energy. Fig.~\ref{Fig_sim_4} shows the time-domain responses of the three-converter system under active power fluctuations (by adding stochastic signals to the active power references of the converters to emulate such effects). The other settings are the same as the settings of Fig.~\ref{Fig_sim_2}. It can be seen from Fig.~\ref{Fig_sim_4} that the stability results are the same as those in Fig.~\ref{Fig_sim_2}, namely, the system is unstable when all the converters employ GFL control while it is stable by changing Converter~3 from GFL control to GFM control. Moreover, compared with Fig.~\ref{Fig_sim_2}~(a), the oscillations in Fig.~\ref{Fig_sim_4}~(a) are further amplified due to the active power fluctuations. Overall, the above time-domain simulation results are aligned with the gain and phase analysis provided in Examples~\ref{ex:three_C} and~\ref{ex:three_C_GFM}.

\begin{figure}[!t]
	\centering
	\includegraphics[width=2.7in]{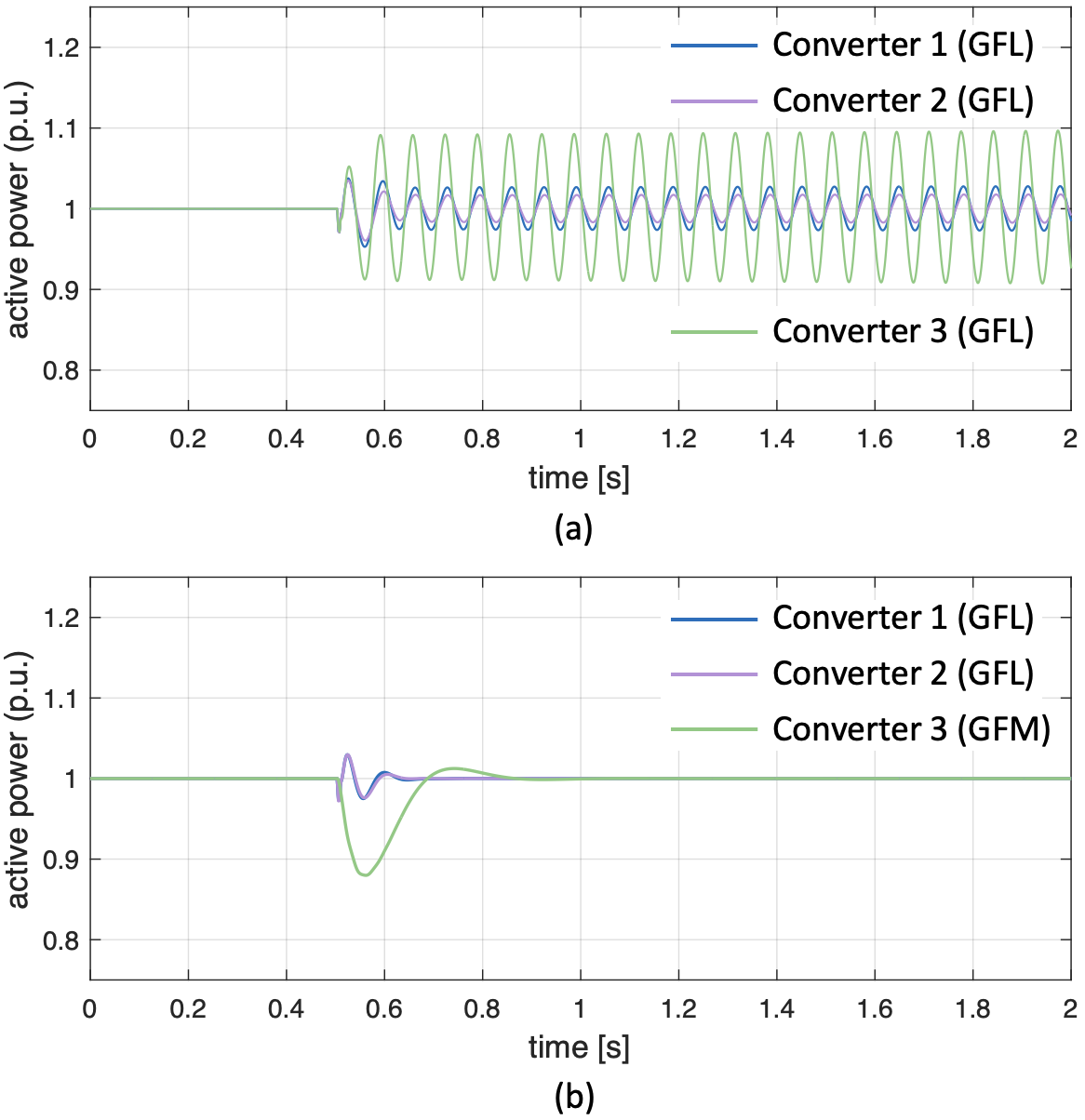}
	\vspace{-3mm}
	\caption{Time-domain responses of the three-converter system in Fig.~\ref{Fig_3_converter}: (a)~all the converters employ GFL control (corresponding to Example~\ref{ex:three_C}), and (b)~Converters~1 and~2 employ GFL control and Converter~3 employs GFM control (corresponding to Example~\ref{ex:three_C_GFM}).}
	\vspace{0mm}
	\label{Fig_sim_2}
\end{figure}

\begin{figure}[!t]
	\centering
	\includegraphics[width=2.7in]{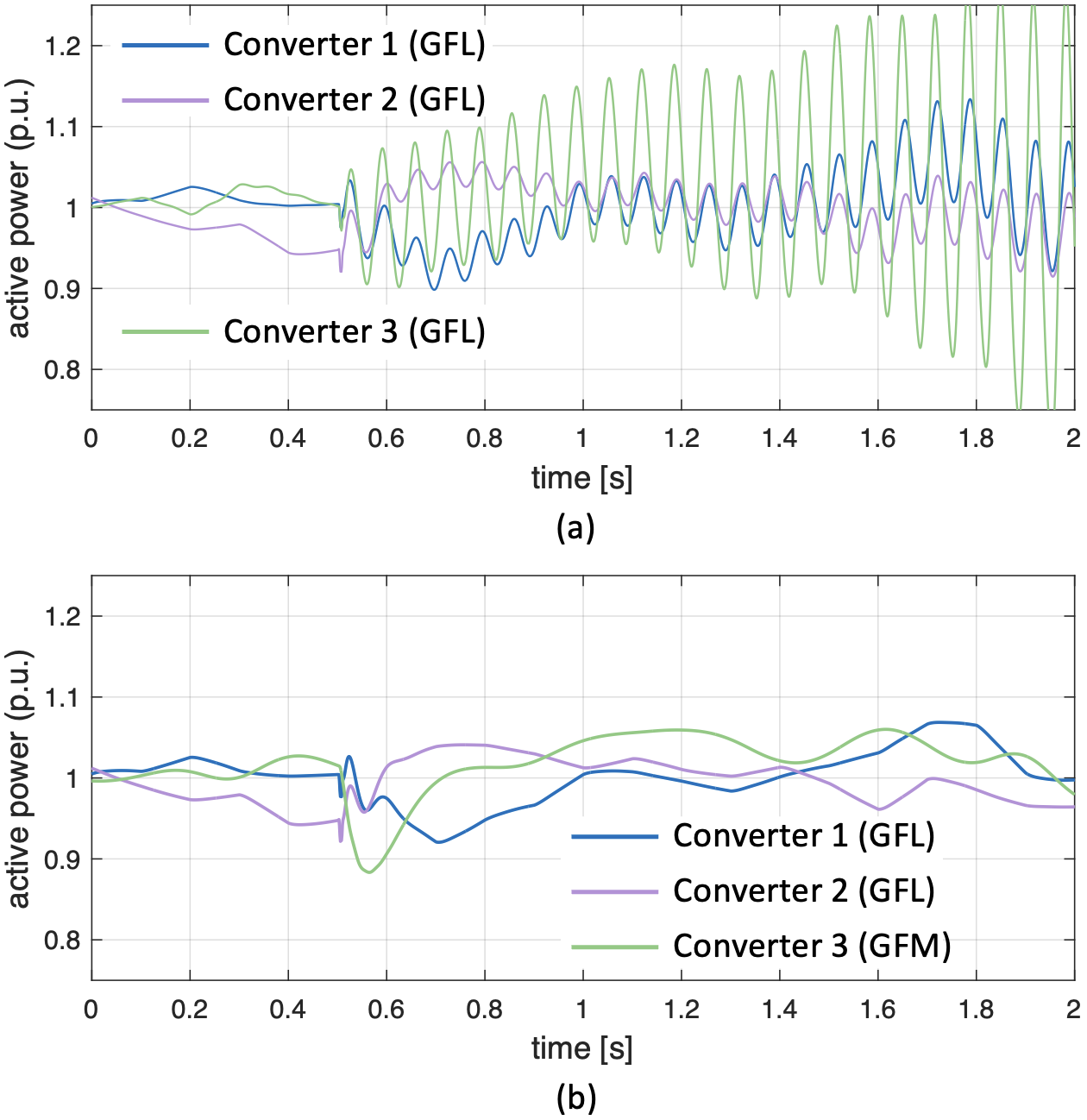}
	\vspace{-3mm}
	\caption{\tr Time-domain responses of the three-converter system in Fig.~\ref{Fig_3_converter} under active power fluctuations: (a)~all the converters employ GFL control (corresponding to Example~\ref{ex:three_C}), and (b)~Converters~1 and~2 employ GFL control and Converter~3 employs GFM control (corresponding to Example~\ref{ex:three_C_GFM}).}
	\vspace{-2mm}
	\label{Fig_sim_4}
\end{figure}

}

{\tr
\subsection{Time-domain simulation results of Example~\ref{ex:68_bus}}

\tr

Fig.~\ref{Fig_68_timedomain}~(a) shows the time-domain responses of the 68-bus system in Fig.~\ref{Fig_68_bus}, where the grid-side impedance of GFL converter~13 steps from $0.013+j0.13~{\rm pu}$ to $0.026+j0.26~{\rm pu}$ at $t=0.2~{\rm s}$ to emulate, e.g., the tripping of one (out of two) transmission line. Note that after $t=0.2~{\rm s}$, the setting of the system is fully aligned with Fig.~\ref{Fig_68_unstable} and Fig.~\ref{Fig_68_unstable_zoomin}~(a) in Example~\ref{ex:68_bus}. It can be seen from Fig.~\ref{Fig_68_timedomain}~(a) that the 68-bus system is unstable under such a setting -- the frequency, voltage magnitude, and active power are all oscillating. 
This is consistent with the analysis in Example~\ref{ex:68_bus} that the high gains of GFL converters~12 and~13 result in the gray area in Fig.~\ref{Fig_68_unstable} where the proposed conditions cannot be satisfied. Moreover, it can be observed that the oscillation frequency is 11.3~Hz, lying exactly in the gray area.

\tr

To alleviate the high gains of GFL converters~12 and~13, we reduce their PLL bandwidths from 60~rad/s and 70~rad/s to 5~rad/s, corresponding to the setting in Fig.~\ref{Fig_68_unstable_zoomin}~(b). Fig.~\ref{Fig_68_timedomain}~(b) shows the time-domain responses of the 68-bus system in Fig.~\ref{Fig_68_bus} under such a setting (the same disturbance occurs at $t=0.2~{\rm s}$). Note that after $t=0.2~{\rm s}$, the setting of the system is fully aligned with Fig.~\ref{Fig_68_unstable_zoomin}~(b) in Example~\ref{ex:68_bus}.
It can be seen that the system is stable, consistent with the results in Fig.~\ref{Fig_68_unstable_zoomin}~(b) showing that the proposed conditions are satisfied at any frequency (after reducing the PLL bandwidths). In summary, the above time-domain simulation results validate the effectiveness of the analysis in the previous examples.

\tr
\begin{figure}[!t]
	\centering
	\includegraphics[width=3.43in]{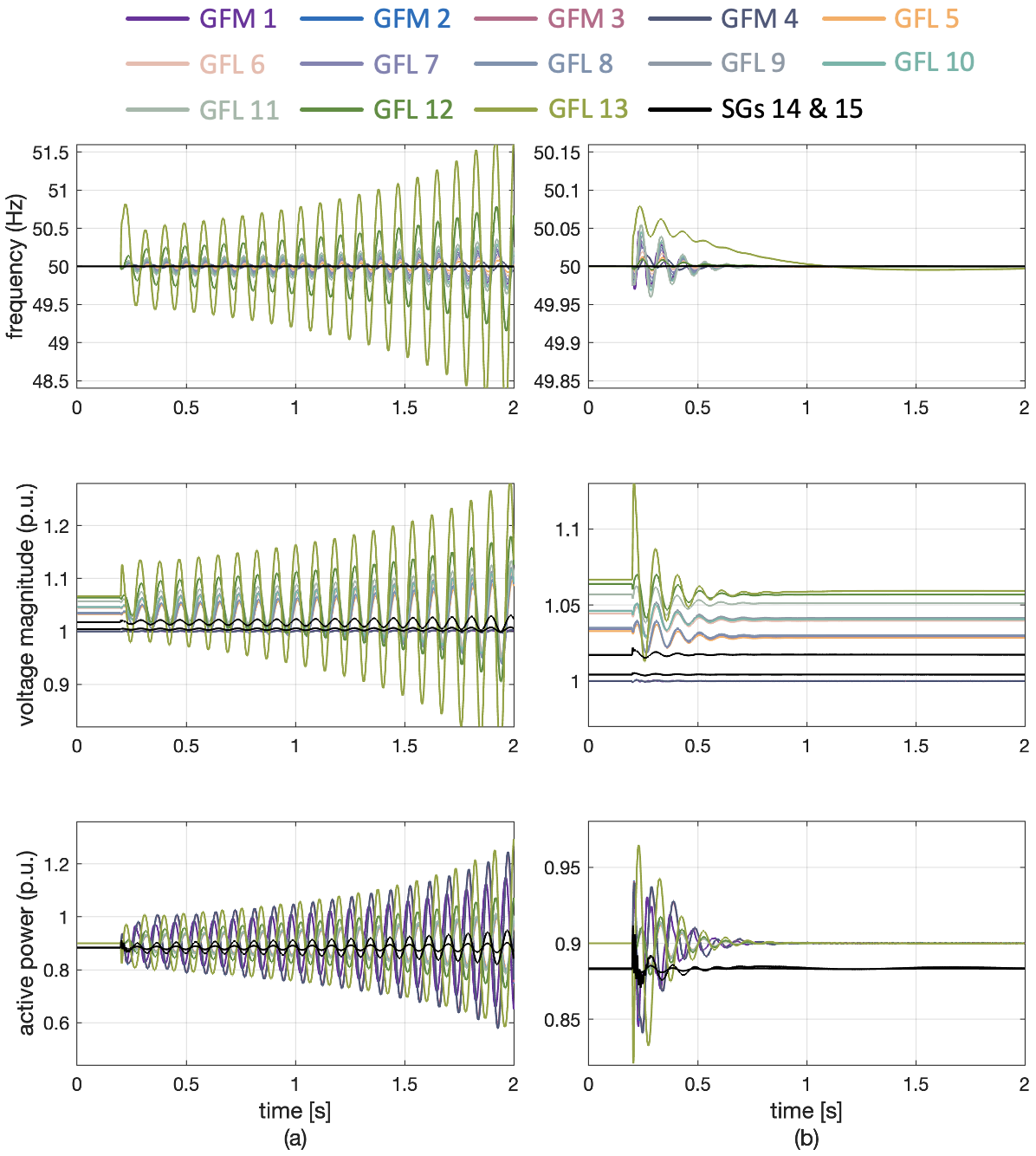}
	\vspace{-3mm}
	\caption{\tr Time-domain responses of the 68-bus system: (a)~using the setting in Fig.~\ref{Fig_68_unstable} and Fig.~\ref{Fig_68_unstable_zoomin}-a; (b)~using the setting in Fig.~\ref{Fig_68_unstable_zoomin}-b where the PLL bandwidths of GFL converters 12 and 13 are reduced.}
	\vspace{0mm}
	\label{Fig_68_timedomain}
\end{figure}



}

        \section{Comparison with Passivity-Based Analysis and Small Gain Analysis}

\begin{figure}[!t]
	\centering
	\includegraphics[width=3in]{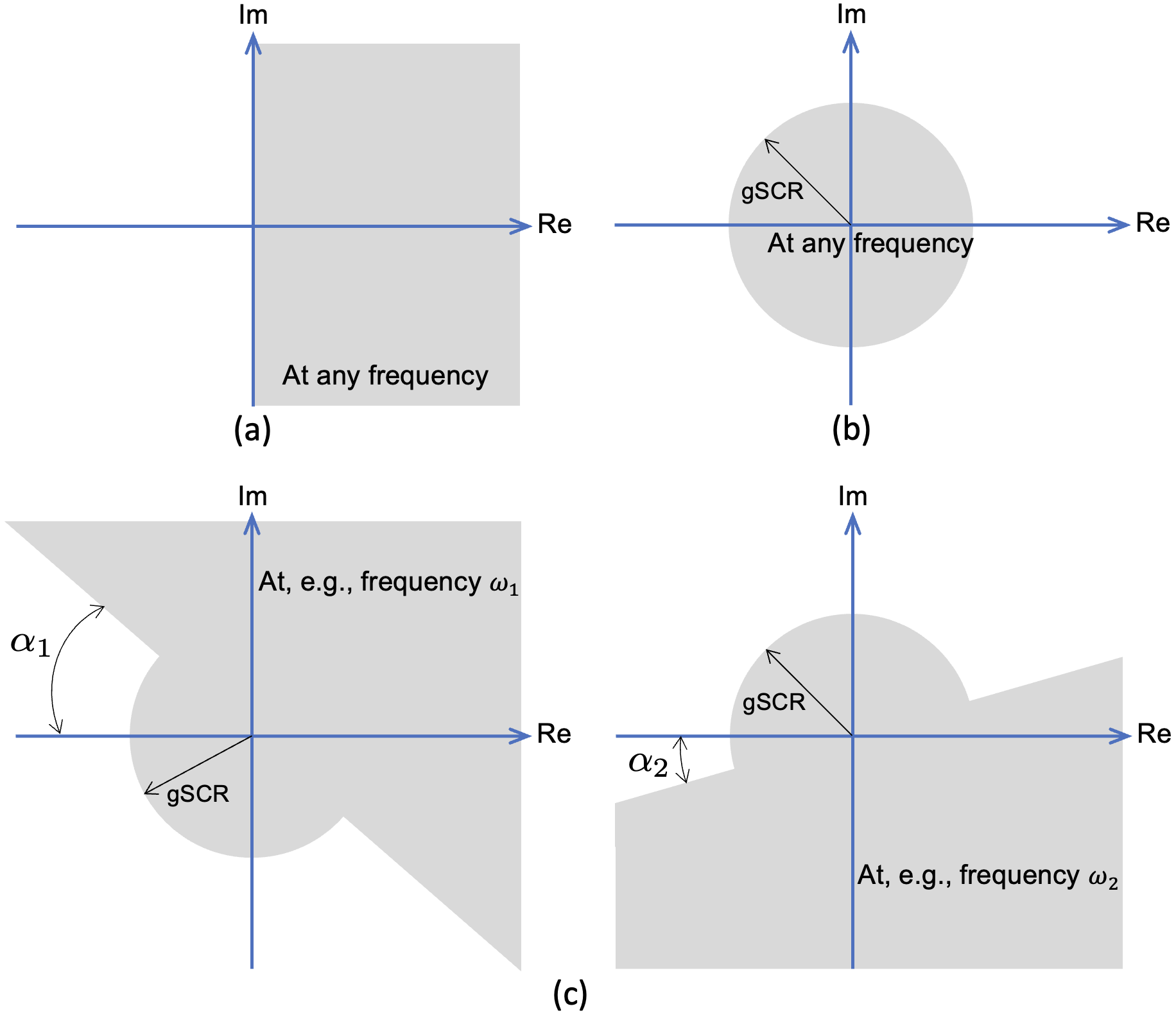}
	\vspace{-3mm}
	\caption{Comparison between passivity-based condition, small gain condition, and Corollary~\ref{prop:Power_grid_strength}: (a)~admissible region of passivity, (b)~admissible region of small gain condition, and (c)~admissible region of Corollary~\ref{prop:Power_grid_strength}.}
	\vspace{-4mm}
	\label{Fig_compare_passivity}
\end{figure}

    Passivity-based analysis and small gain theorem-based analysis also provide decentralized stability conditions for multi-converter systems~\cite{harnefors2015passivity, huang2020h}. The passivity-based condition requires that each converter is passive, which is equivalent to constraining the phases of $\widetilde{\bf Y}_{{\rm C},i}(j\omega)$ within $[-90^\circ,+90^\circ]$ \textit{at any frequency}, as illustrated in the grey region of Fig.~\ref{Fig_compare_passivity}~(a). The small gain condition instead requires that the gains of $\widetilde{\bf Y}_{{\rm C},i}(j\omega)$ are within the grey circle in Fig.~\ref{Fig_compare_passivity}~(b) \textit{at any frequency}.
    However, such conditions may not be satisfied due to its conservatism. By comparison, our conditions in this paper are much less conservative. For instance, Corollary~\ref{prop:Power_grid_strength} requires that \textit{at a certain frequency}, either the gain of $\widetilde{\bf Y}_{{\rm C},i}(j\omega)$ is less than gSCR or its phases are contained in a half-plane with a certain angle (e.g., at frequency $\omega_1$, the angle of the half-plane can be $\alpha_1$, while at frequency $\omega_2$, the angle can be $\alpha_2$, as shown in Fig.~\ref{Fig_compare_passivity}~(c)). It can be seen from Fig.~\ref{Fig_compare_passivity} that our conditions provide more flexibility and are therefore less conservative than conventional passivity-based analysis or small gain theorem-based analysis.
	
        \section{Conclusions}
	
    This paper proposed decentralized stability conditions for multi-converter systems, which evaluate the closed-loop stability by comparing the gain curves and phase areas of the converters with the gain curve and phase area of the power network. {\tr The proposed conditions can be checked in different converters in a decentralized manner and thus can be used in large-scale heterogeneous multi-converter systems. Note that in this paper, we tested our approach in a 9-bus system and a 68-bus system, while it can also be applied to larger systems, as the conditions are checked by independently plotting the network dynamics and the converters' dynamics in terms of gain and phase.}
    Moreover, we found that to reduce conservatism, the GFM dynamics should be included in the network dynamics when splitting the GFM-GFL hybrid multi-converter system. This is consistent with the intuition that GFM converters behave as voltage sources and enhance the power grid strength (e.g., gSCR), which can be considered as a dominant part of the grid.
    Our approach intuitively displays the interaction between the converters' dynamics and the network dynamics and gives guidelines on how the converters' characteristics should match the network dynamics to enforce stability, which is useful in the control design of converters to ensure interoperability. 
    
    Future works should include {\tr approximation methods to reduce the required knowledge in the proposed decentralized conditions as well as} systematic control design methods based on the decentralized gain and phase conditions. {\tr Recent advances in system analysis have shown that both the small gain theorem and the small phase theorem have the potential to analyze the large signal stability problems of nonlinear systems, e.g., transient stability problems in power systems, which can be future work. Moreover, we have observed that the sectorial transition frequency of converters plays a significant role in the stability performance of the system, and future work can also include an in-depth analysis of how this frequency is affected by, for instance, parameters and control structure of converters.}

        \normalem
        \bibliographystyle{IEEEtran}
	\bibliography{references}

%
%

        
        \vspace{15mm}
        
	\appendices

        \section{Proof of Proposition~\ref{prop:Decentralized_stability_conditions}}
        \label{Append_Proof1}

        \begin{proof}
It can be deduced from~\eqref{eq:admittance_devices} and~\eqref{eq:Y_grid} that the open-loop transfer function of the multi-converter system is 
\begin{equation*}
\begin{split}
L(s) =& \; ({\bf S}\otimes I_2) e^{J {\vartheta}} \; {\bf Y}_{\rm C}^N(s) \; e^{-J {\vartheta}} \; {\bf Y}^{-1}_{\rm grid}(s) \\
=& \; ({\bf S}^\frac{1}{2} \otimes I_2) e^{J {\vartheta}} \; {\bf Y}_{\rm C}^N(s) \; e^{-J {\vartheta}} \; ({\bf S}^\frac{1}{2} \otimes I_2) {\bf Y}^{-1}_{\rm grid}(s) \\
=& \; ({\bf S}^\frac{1}{2} \otimes I_2) e^{J {\vartheta}} \; {\bf Y}_{\rm C}^N(s) \; e^{-J {\vartheta}} \; ({\bf D} \otimes {\widetilde F}_\epsilon^{-1}(s)) \\
& \; ({\bf D}^{-1}{\bf S}^\frac{1}{2} \otimes {\widetilde F}_\epsilon(s)) {\bf Y}^{-1}_{\rm grid}(s). \\
\end{split}
\end{equation*}
Hence, the closed-loop system dynamics can be described by $[({\bf S}\otimes I_2)e^{J {\vartheta}} \; {\bf Y}_{\rm C}^N(s) \; e^{-J {\vartheta}}]\; \# \;  {\bf Y}^{-1}_{\rm grid}(s)$, or equivalently, by~\eqref{eq:rescaled_sys}.
For such an interconnected system,
Condition i) in Lemma~\ref{small-gain-phase} requires that at the given frequency $\omega\in [0,\infty]$, 
\begin{equation}\label{eq:singular_value1}
\overline \sigma(e^{J {\vartheta}} \; {\bf Y}_{\rm C}^N(j\omega) \; e^{-J {\vartheta}} ({\bf D} \otimes {\widetilde F}_\epsilon^{-1}(j\omega)))  \overline \sigma( \widetilde{\bf Y}_{\rm grid}^{-1}(j\omega) ) < 1.
\end{equation}
Notice that $e^{J {\vartheta}} \; {\bf Y}_{\rm C}^N(j\omega) \; e^{-J {\vartheta}} ({\bf D} \otimes {\widetilde F}_\epsilon^{-1}(j\omega))$ is a block-diagonal matrix (with $N$ $2 \times 2$ blocks) with the $i$-th block being $D_i e^{J\theta_i} {\bf Y}_{{\rm C},i}(j\omega) e^{-J\theta_i} {\widetilde F}_\epsilon^{-1}(j\omega)$. Hence, we have
\begin{equation*}
\begin{split}
& \overline \sigma(e^{J {\vartheta}} \; {\bf Y}_{\rm C}^N(j\omega) \; e^{-J {\vartheta}} ({\bf D} \otimes {\widetilde F}_\epsilon^{-1}(j\omega))) \\
= \; & \max\limits_{i} \; \overline \sigma(D_i e^{J\theta_i} {\bf Y}_{{\rm C},i}(j\omega) e^{-J\theta_i} {\widetilde F}_\epsilon^{-1}(j\omega)) \\
= \; & \max\limits_{i} \; \overline \sigma(D_i e^{J\theta_i} {\bf Y}_{{\rm C},i}(j\omega) {\widetilde F}_\epsilon^{-1}(j\omega) e^{-J\theta_i} ) \\
= \; & \max\limits_{i} \; \overline \sigma( D_i {\bf Y}_{{\rm C},i}(j\omega) {\widetilde F}_\epsilon^{-1}(j\omega) ) 
= \;  \max\limits_{i} \; \overline \sigma( \widetilde{\bf Y}_{{\rm C},i}(j\omega) ),
\end{split}
\end{equation*}
where the first equality holds thanks to the block-diagonal structure, the second because of the fact that ${\widetilde F}_\epsilon^{-1}(j\omega)$ and $e^{-J\theta_i}$ have the special structure $\scriptsize \begin{bmatrix} a & -b \\ b & a \end{bmatrix}$ and thus their multiplication is commutative, the third equality holds since $e^{J\theta_i}$ is a unitary matrix, and the fourth equality follows from the definition of $\widetilde{\bf Y}_{{\rm C},i}(s)$ in~\eqref{eq:Y_C_convert}. Then, we rewrite~\eqref{eq:singular_value1} as
\begin{equation*}
\begin{split}
& \left[ \max\limits_{i} \; \overline \sigma( \widetilde{\bf Y}_{{\rm C},i}(j\omega) ) \right] \overline \sigma( \widetilde{\bf Y}_{\rm grid}^{-1}(j\omega) ) \\
= \; & \left[ \max\limits_{i} \; \overline \sigma( \widetilde{\bf Y}_{{\rm C},i}(j\omega) ) \right] \frac{1}{\underline \sigma( \widetilde{\bf Y}_{\rm grid}(j\omega) )} <1,
\end{split}
\end{equation*}
which further leads to~\eqref{eq:decentralized_gain}.

For the phase condition in ii) of Lemma~\ref{small-gain-phase}, we need $e^{J {\vartheta}} \; {\bf Y}_{\rm C}^N(j\omega) \; e^{-J {\vartheta}} ({\bf D} \otimes {\widetilde F}_\epsilon^{-1}(j\omega))$ and $\widetilde{\bf Y}_{\rm grid}^{-1}(j\omega)$ to be sectorial at the given frequency $\omega\in [0,\infty]$, while satisfying
\begin{equation}\label{eq:phase1}
\overline \phi(e^{J {\vartheta}} \; {\bf Y}_{\rm C}^N(j\omega) \; e^{-J {\vartheta}} ({\bf D} \otimes {\widetilde F}_\epsilon^{-1}(j\omega))) +  \overline \phi( \widetilde{\bf Y}_{\rm grid}^{-1}(j\omega) ) < \pi, \vspace{-1mm}
\end{equation}
and
\vspace{-1mm}
\begin{equation}\label{eq:phase2}
\underline \phi(e^{J {\vartheta}} \; {\bf Y}_{\rm C}^N(j\omega) \; e^{-J {\vartheta}} ({\bf D} \otimes {\widetilde F}_\epsilon^{-1}(j\omega))) +  \underline \phi( \widetilde{\bf Y}_{\rm grid}^{-1}(j\omega) ) > - \pi.
\end{equation}
Again, since $e^{J {\vartheta}} \; {\bf Y}_{\rm C}^N(j\omega) \; e^{-J {\vartheta}} ({\bf D} \otimes {\widetilde F}_\epsilon^{-1}(j\omega))$ is block-diagonal, we have
\begin{align}\label{eq:max_phase}
&\overline \phi(e^{J {\vartheta}} \; {\bf Y}_{\rm C}^N(j\omega) \; e^{-J {\vartheta}} ({\bf D} \otimes {\widetilde F}_\epsilon^{-1}(j\omega))) \nonumber\\
= \; & \max\limits_{i} \; \overline \phi( D_i e^{J\theta_i} {\bf Y}_{{\rm C},i}(j\omega) e^{-J\theta_i} {\widetilde F}_\epsilon^{-1}(j\omega) )\\
= \; & \max\limits_{i} \; \overline \phi( D_i e^{J\theta_i} {\bf Y}_{{\rm C},i}(j\omega) {\widetilde F}_\epsilon^{-1}(j\omega) e^{-J\theta_i}  )\nonumber\\
= \; & \max\limits_{i} \; \overline \phi( D_i {\bf Y}_{{\rm C},i}(j\omega) {\widetilde F}_\epsilon^{-1}(j\omega) ) 
=  \max\limits_{i} \; \overline \phi( \widetilde{\bf Y}_{{\rm C},i}(j\omega) ),\nonumber
\end{align}
and one can analogously obtain
\begin{equation}\label{eq:min_phase}
\underline \phi(e^{J {\vartheta}} \; {\bf Y}_{\rm C}^N(j\omega) \; e^{-J {\vartheta}} ({\bf D} \otimes {\widetilde F}_\epsilon^{-1}(j\omega))) 
= \min\limits_{i} \; \underline \phi( \widetilde{\bf Y}_{{\rm C},i}(j\omega) ).
\end{equation}
By substituting~\eqref{eq:max_phase} and~\eqref{eq:min_phase} into~\eqref{eq:phase1} and~\eqref{eq:phase2}, we obtain a) and b) in~\eqref{eq:decentralized_phase}. We also need c) in~\eqref{eq:decentralized_phase} to ensure that the block-diagonal matrix $e^{J {\vartheta}} \; {\bf Y}_{\rm C}^N(j\omega) \; e^{-J {\vartheta}} ({\bf D} \otimes {\widetilde F}_\epsilon^{-1}(j\omega))$ remains sectorial. This completes the proof.
\end{proof}

\section{Proof of Corollary~\ref{prop:Power_grid_strength}}
        \label{Append_Proof2}

        \begin{proof}
The result is a special case of Proposition~\ref{prop:Decentralized_stability_conditions}, as shown in what follows.
Since the network dynamics can be represented by~\eqref{eq:Ygrid_RX}, by choosing ${\bf D} = I_N$, we have
\begin{equation}\label{eq:const_Ygrid}
\begin{split}
\widetilde{\bf Y}_{\rm grid}(s) & = ({\bf S}^{-\frac{1}{2}} \otimes I_2) {\bf Y}_{\rm grid}(s) ( {\bf S}^{-\frac{1}{2}} {\bf D} \otimes {F}_\epsilon^{-1}(s)) \\
& = {\bf S}^{-\frac{1}{2}} {\bf B}_{\rm r} {\bf S}^{-\frac{1}{2}} \otimes I_2 .\\
\end{split}
\end{equation}
Further, we have $\underline \sigma(\widetilde{\bf Y}_{\rm grid}(j\omega)) = \underline \sigma({\bf S}^{-\frac{1}{2}} {\bf B}_{\rm r} {\bf S}^{-\frac{1}{2}})$. Notice that ${\bf S}^{-\frac{1}{2}} {\bf B}_{\rm r} {\bf S}^{-\frac{1}{2}}$ is a positive definite matrix and similar to ${\bf S}^{-1} {\bf B}_{\rm r} $, so $\underline \sigma({\bf S}^{-\frac{1}{2}} {\bf B}_{\rm r} {\bf S}^{-\frac{1}{2}}) = \lambda_1({\bf S}^{-\frac{1}{2}} {\bf B}_{\rm r} {\bf S}^{-\frac{1}{2}}) = \lambda_1({\bf S}^{-1} {\bf B}_{\rm r})$. As defined in~\cite{zhang2018assessing} and~\cite{dong2018small}, $\lambda_1({\bf S}^{-1} {\bf B}_{\rm r})$ is the generalized short-circuit ratio (gSCR) of a multi-converter system that reflects the power grid strength and also the (weighted) power network connectivity~\cite{huang2022impacts, dorfler2012kron}. Based on the above, we derive that the condition in~\eqref{eq:decentralized_gain} reduces to~\eqref{eq:decentralized_gain_gSCR}. 
According to~\eqref{eq:const_Ygrid}, $\widetilde{\bf Y}_{\rm grid}(j\omega)$ is a (constant) positive definite matrix. Hence, we have $\overline \phi( \widetilde{\bf Y}^{-1}_{\rm grid}(j\omega) = \underline \phi( \widetilde{\bf Y}^{-1}_{\rm grid}(j\omega) = 0$ and obtain~\eqref{eq:decentralized_phase_gSCR} based on the condition in~\eqref{eq:decentralized_phase}. This completes the proof.
\end{proof}

\section{Proof of Lemma~\ref{lem:Equivalent_systems}}
        \label{Append_Proof3}

\begin{proof}
The characteristic polynomial of the systems is 
\begin{equation}\label{eq:char_poly}
{\rm det}( ({\bf S}\otimes I_2)e^{J {\vartheta}} \; {\bf Y}_{\rm C}^N(s) \; e^{-J {\vartheta}}\; {\bf Y}^{-1}_{\rm grid}(s) +I ) = 0,
\end{equation}
which, by using Schur complement, can be rewritten as
\begin{align}\label{eq:char_poly_equil}
& {\rm det}( ({\bf S}\otimes I_2)e^{J {\vartheta}} \; {\bf Y}_{\rm C}^N(s) \; e^{-J {\vartheta}}\; {\bf Y}^{-1}_{\rm grid}(s) +I ) = 0 \nonumber\\
= \; & {\rm det}( ({\bf S}\otimes I_2)e^{J {\vartheta}} \; {\bf Y}_{\rm C}^N(s) \; e^{-J {\vartheta}} + {\bf Y}_{\rm grid}(s) ) \; {\rm det}( {\bf Y}^{-1}_{\rm grid}(s) ) \nonumber\\
= \; & {\rm det}\left( \begin{bmatrix} {\bf Y}_{{\rm C}\{1\}}(s) & \\ & {\bf Y}_{{\rm C}\{2\}}(s) \end{bmatrix} + \begin{bmatrix}
    {\bf Y}_{\rm g}^1(s) & {\bf Y}_{\rm g}^2(s) \vspace{1mm} \\ 
    {\bf Y}_{\rm g}^3(s) & {\bf Y}_{\rm g}^4(s) \\
\end{bmatrix} \right) \nonumber\\
& {\rm det}( {\bf Y}^{-1}_{\rm grid}(s) ) \nonumber\\
= \; & {\rm det}( {\bf Y}_{{\rm C}\{2\}}(s) + {\bf Y}_{\rm g}^4(s) ) \; {\rm det}( {\bf Y}_{{\rm C}\{1\}}(s) + {\bf Y}_{\rm gridC}(s) ) \nonumber\\
& {\rm det}( {\bf Y}^{-1}_{\rm grid}(s) ) \nonumber\\
= \; & {\rm det}( {\bf Y}_{{\rm C}\{2\}}(s) + {\bf Y}_{\rm g}^4(s) ) \; {\rm det}( {\bf Y}_{{\rm C}\{1\}}(s)  {\bf Y}^{-1}_{\rm gridC}(s) +I ) \nonumber\\
& {\rm det}( {\bf Y}_{\rm gridC}(s) ) \; {\rm det}( {\bf Y}^{-1}_{\rm grid}(s) ) \nonumber\\
= \; & {\rm det}( {\bf Y}_{{\rm C}\{1\}}(s)  {\bf Y}^{-1}_{\rm gridC}(s) +I )  \nonumber\\
& {\rm det}\left( \begin{bmatrix} 0 & \\ & {\bf Y}_{{\rm C}\{2\}}(s) \end{bmatrix} {\bf Y}^{-1}_{\rm grid}(s) +I \right).
\end{align}
The above derivation shows that the characteristic polynomial in~\eqref{eq:char_poly} can be rewritten as the product of the characteristic polynomials of subsystem~1 in~\eqref{eq:subsys1} and subsystem~2 in~\eqref{eq:subsys2}. Hence, the original multi-converter system is stable if subsystem~1 in~\eqref{eq:subsys1} and subsystem~2 in~\eqref{eq:subsys2} are stable. 

Moreover, it holds that
\begin{equation*}
\begin{split}
& {\rm det}\left( \begin{bmatrix} 0 & \\ & {\bf Y}_{{\rm C}\{2\}}(s) \end{bmatrix} {\bf Y}^{-1}_{\rm grid}(s) +I \right) \\
= \; & {\rm det}\left( \begin{bmatrix} 0 & \\ & {\bf Y}_{{\rm C}\{2\}}(s) \end{bmatrix} + {\bf Y}_{\rm grid}(s) \right) \; {\rm det}( {\bf Y}^{-1}_{\rm grid}(s) ) \\
= \; & {\rm det}( {\bf Y}_{\rm gridC}(s) ) \; {\rm det}( {\bf Y}_{{\rm C}\{2\}}(s) + {\bf Y}_{\rm g}^4(s) ) \; {\rm det}( {\bf Y}^{-1}_{\rm grid}(s) ) ,
\end{split}
\end{equation*}
which indicates that if subsystem~1 has no unstable poles, then ${\bf Y}_{\rm gridC}(s)$ has no unstable zeros and therefore ${\bf Y}^{-1}_{\rm gridC}(s)$ has no unstable poles.
This completes the proof.
\end{proof}

        \section{Parameters of the Test Systems}
        \label{Appenxix:GFL_parameters}

        The main parameters of the GFL converters are: $L_F = 0.05~{\rm pu}$, $C_F = 0.06~{\rm pu}$, $L_g = 0.15~{\rm pu}$ in Examples~\ref{ex:single_C},~\ref{ex:three_C}, and~\ref{ex:three_C_GFM} (while in Example~\ref{ex:68_bus}, $L_g = 0.2~{\rm pu}$ for GFL converters 5$\sim$11 and $L_g = 0.26~{\rm pu}$ for GFL converters 12 and 13), PI parameters of current control loop: $\{0.3,10\}$, time constant for voltage feedforward: $0.02$, PI parameters of active/reactive power control loop: $\{0.5,40\}$, PLL bandwidth: 40~rad/s (unless otherwise specified).

        The main parameters of the GFM converters are (the control scheme is similar to the one in~\cite{xin2022many}): $L_F = 0.05~{\rm pu}$, $C_F = 0.06~{\rm pu}$, $L_g = 0.15~{\rm pu}$ in Example~\ref{ex:three_C_GFM} (while in Example~\ref{ex:68_bus}, $L_g = 0.12~{\rm pu}$ for GFM converter 1 and $L_g = 0.1~{\rm pu}$ for GFM converters~2$\sim$4), PI parameters of current control loop: $\{0.3,10\}$, time constant for voltage feedforward: $0.02$, PI parameters of AC voltage control loop (with output current feedforward): $\{2,10\}$, inertia and damping coefficients: $\{2,50\}$ in Example~\ref{ex:three_C_GFM} ($\{0.2,50\}$ for GFM converter~1 and $\{4,50\}$ for GFM converters~2$\sim$4 in Example~\ref{ex:68_bus}).

        The main parameters of the SGs are (the notations are the same as those in~\cite{huang2020damping}): $J_{SG} = 10.39~{\rm pu}$, $X_d = 1.81~{\rm pu}$, $X_q = 1.76~{\rm pu}$, $X'_d = 0.3~{\rm pu}$, $X'_q = 0.65~{\rm pu}$, $X''_d = 0.23~{\rm pu}$, $X''_q = 0.25~{\rm pu}$, $T'_d = 8~{\rm pu}$, $T'_q = 1~{\rm pu}$, $T''_d = 0.03~{\rm pu}$, $T''_q = 0.07~{\rm pu}$; Governor parameters (standard IEEEG1 model): $K_{\rm droop} = 8~{\rm pu}$, $T_1 = 0.5$, $T_2 = 1$, $T_3 = 0.6$, $T_4 = 0.6$, $T_5 = 0.5$, $T_6 = 0.8$, $T_7 = 1$, $K = 5$, $K_1 = 0.3$, $K_3 = 0.25$, $K_5 = 0.3$, $K_7 = 0.15$; Fast excitor parameters (standard IEEET1 model): $K_a = 15$, $T_a = 0.05$, $K_f = 0.0057$, $T_f = 0.5$, $T_r = 0.1$.

        The main power network parameters of the modified 68-bus system are as follows. The load profile, shunt capacitors (in the Pi transmission line model), and capacities of the generators are the same as those in~\cite{canizares2016benchmark}. The base value for power is 100~MVA and for frequency it is 50~Hz. Impedances ($Z_{i,j} = R_{i,j}+jX_{i,j}$) of the transmission lines ($\times 10^{-2}~{\rm pu}$): 
        $Z_{1,54} =j0.0905$, $Z_{2,58} =j0.1250$, $Z_{3,62} =j0.1$, $Z_{4,19} =0.0035+j0.071$, $Z_{5,20} =0.0045+j0.09$, $Z_{6,22} =j0.0715$, $Z_{7,23} =0.0025+j0.136$, $Z_{8,25} =0.003+j0.116$, $Z_{9,29} =0.004+j0.078$, $Z_{10,31} =j0.13$, $Z_{11,32} =j0.065$, $Z_{12,36} =j0.0375$, $Z_{13,17} =j0.2475$, $Z_{14,41} =j0.0075$, $Z_{15,42} =j0.0075$, $Z_{16,18} =j0.015$, $Z_{17,36} =0.0025+j0.0225$, $Z_{17,43} =0.0025+j0.138$, $Z_{18,42} =0.002+j0.03$, $Z_{18,49} =0.038+j0.5709$, $Z_{18,50} =0.006+j0.144$, $Z_{19,20} =0.0035+j0.069$, $Z_{19,68} =0.008+j0.0976$, $Z_{21,22} = 0.004+j0.07$, $Z_{21,68} = 0.004+j0.0675$, $Z_{22,23} = 0.003+j0.048$, $Z_{23,24} = 0.011+j0.175$, $Z_{24,68} = 0.0015+j0.0295$, $Z_{25,26} = 0.016+j0.1615$, $Z_{25,54} = 0.035+j0.043$, $Z_{26,27} = 0.007+j0.0735$, $Z_{26,28} = 0.0215+j0.237$, $Z_{26,29} = 0.0285+j0.3125$, $Z_{27,37} = 0.0065+j0.0865$, $Z_{27,53} = 0.16+j1.6$, $Z_{28,29} = 0.007+j0.0755$, $Z_{30,31} = 0.0065+j0.0935$, $Z_{30,32} = 0.012+j0.144$, $Z_{30,53} = 0.004+j0.037$, $Z_{30,61} = 0.0047+j0.0458$, $Z_{31,38} = 0.0055+j0.0735$, $Z_{31,53} = 0.008+j0.0815$, $Z_{32,33} = 0.004+j0.0495$, $Z_{33,34} = 0.0055+j0.0785$, $Z_{33,38} = 0.018+j0.222$, $Z_{34,35} = 0.0005+j0.037$, $Z_{34,36} = 0.0165+j0.0555$, $Z_{35,45} = 0.0035+j0.0875$, $Z_{36,61} = 0.0055+j0.049$, $Z_{37,52} = 0.0035+j0.041$, $Z_{37,68} = 0.0035+j0.0445$, $Z_{38,46} = 0.011+j0.142$, $Z_{39,44} = j0.2055$, $Z_{39,45} = j0.4195$, $Z_{40,41} = 0.03+j0.42$, $Z_{40,48} = 0.01+j0.11$, $Z_{41,42} = 0.02+j0.3$, $Z_{43,44} = 0.0005+j0.0055$, $Z_{44,45} = 0.0125+j0.365$, $Z_{45,51} = 0.002+j0.0525$, $Z_{46,49} = 0.009+j0.137$, $Z_{47,48} = 0.0063+j0.067$, $Z_{47,53} = 0.0065+j0.094$, $Z_{50,51} = 0.0045+j0.1105$, $Z_{52,55} = 0.0055+j0.0665$, $Z_{53,54} = 0.0175+j0.2055$, $Z_{54,55} = 0.0065+j0.0755$, $Z_{55,56} = 0.0065+j0.1065$, $Z_{56,57} = 0.004+j0.064$, $Z_{56,66} = 0.004+j0.0645$, $Z_{57,58} = 0.001+j0.013$, $Z_{57,60} = 0.004+j0.056$, $Z_{58,59} = 0.003+j0.046$, $Z_{58,63} = 0.0035+j0.041$, $Z_{59,60} = 0.002+j0.023$, $Z_{60,61} = 0.0115+j0.1815$, $Z_{62,63} = 0.002+j0.0215$, $Z_{62,65} = 0.002+j0.0215$, $Z_{63,64} = 0.008+j0.2175$, $Z_{64,65} = 0.008+j0.2175$, $Z_{65,66} = 0.0045+j0.0505$, $Z_{66,67} = 0.009+j0.1085$, $Z_{67,68} = 0.0045+j0.047$.


\end{document}